%% file: main.tex
\documentclass[letterpaper,twocolumn]{article}
\usepackage[textwidth=510pt,textheight=680pt]{geometry}

\input{submission/preambles}
\input{preamble/packages}
\input{preamble/definitions}
\input{preamble/acronyms}

\usepackage[switch]{lineno} 

\begin{document}
%\linenumbers  

\title{Safety Synthesis Sans Specification}%~(\SFour)}
\input{submission/authors}
\date{}
\maketitle

\begin{abstract}
\begin{quote}
\input{submission/abstract}
\end{quote}
\end{abstract}

\input{sections/introduction.tex}

\input{sections/formalization.tex}

\section*{\large Acknowledgments}
\vspace{-3mm}
Authors are grateful to Martin Tappler for his part in this work's feasibility studies and in implementing the first proof of concept. This work was supported by the Austrian Research Promotion Agency (FFG) through project TRUSTED (867558).

\bibliographystyle{plain}
\bibliography{bibfile}
%\end{document}
\clearpage
\newpage
\appendix
\section{Omitted Algorithms}\label{app:algs}
\input{sections/algorithms/rebase.tex}

\section{Omitted Proofs}
\input{sections/omitted_proofs.tex}
\newpage

\section{Running Examples}
Below we provide the details of run of the algorithm on the examples
provided in the body of the paper. In particular, we provide
the intermediate tables, the conjectured symbolic automata, and 
the received counterexamples.

\subsection{Example~\ref{example1:hyp:0}}
Our first example is provided in the introduction under \emph{An illustrative example}.
Recall that it considers the unknown language $U$ consisting of behaviors that grant requests 
either in the step where the  request was received or in the next step. The input variable is $r$ (request), and the single output
variable is $g$ (grant). Recall also that as explained there, there exist an infinite number of concrete transducers realizing $U$.

\begin{figure}[!ht]
    \centering
    \includegraphics[scale=1,page=2]{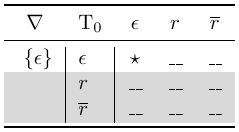}
    \caption{The initial table, where $R$ and $C$ are $\{\epsilon, r, \overline{r}\}$, and $ B=\{ \epsilon \}$. For further readability we shaded rows not in $B$.}
    \label{example1:table:0}
\end{figure}

\begin{figure}[!ht]
    \centering
    \includegraphics[page=3]{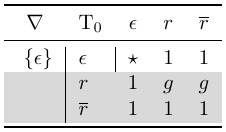}
    \caption{Subsequent filled table.}
    \label{example1:table:1}
\end{figure}

\begin{figure}[!ht]
    \centering
    \includegraphics[page=5]{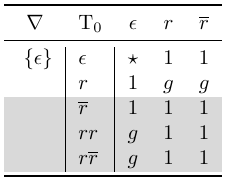}
    \caption{The subsequent table after closing. Since no row in $B$ covers rows $rr$ and $r\overline{r}$, this table is not yet closed.}
    \label{example1:table:2}
\end{figure}

\begin{figure}[!ht]
    \centering
    \includegraphics[page=7]{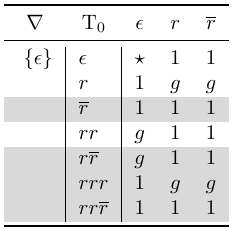}
    \caption{Subsequent closed table. 
    All rows in $R$ are now covered by a row in $B$.}
    \label{example1:table:3}
\end{figure}

\begin{figure}[!ht]
    \centering
    \includegraphics[page=8]{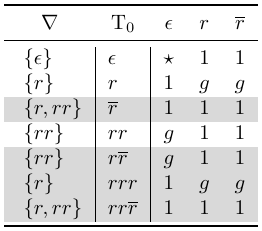}
    \caption{Subsequent rebased table. 
    From this table we can extract the symbolic transducer demonstrated in Fig.\ref{example1:hyp:0}.}
    \label{example1:table:4}
\end{figure}

\newpage

\subsection{Example~\ref{example2:hyp:0}}
Our second example is also provided in the introduction under \emph{An illustrative example}.
In this example, in addition to the previous description, the language disallows two subsequent grants.
%(and hence, in particular, there are never two subsequent requests). 
In this case, there exist two
conflicting implementations realizing the language (as explained there) therefore
we will see a call to \emph{Split} in the course of running the example.

\begin{figure}[!ht]
    \centering
    \includegraphics[scale=1,page=2]{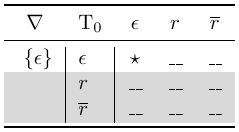}
    \caption{The initial table $T_0$, where $R$ and $C$ are $\{\epsilon, r, \overline{r}\}$, and $ B=\{ \epsilon \}$. We again shade rows that are not in $B$.}
    \label{example2:table:0}
\end{figure}

\begin{figure}[!ht]
    \centering
    \includegraphics[page=3]{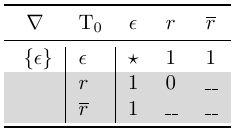}
    \caption{By allowing both $g$, and $\overline{g}$ in $M(\epsilon, r)$, entry $M(r,r)$ becomes empty. 
    Thus, splitting $T_0$ is necessary.}
    \label{example2:table:1}
\end{figure}

\begin{figure}[!ht]
    \centering
    \resizebox{\columnwidth}{!}
    {
      \begin{tabular}{c|c}
        \includegraphics[page=4]{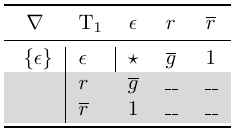} &
        \includegraphics[page=12]{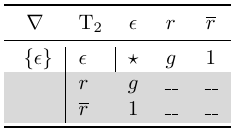}
      \end{tabular}
    }
    \caption{Splitting $T_0$ results in $T_1$ and $T_2$; both of which, we must explore as of now.}
    \label{example2:table:2}
\end{figure}

\begin{figure}[!ht]
    \centering
    \resizebox{\columnwidth}{!}
    {
      \begin{tabular}{c|c}
        \includegraphics[page=5]{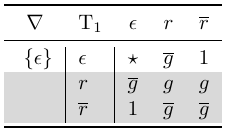} &
        \includegraphics[page=13]{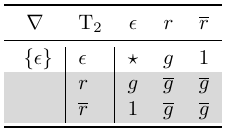}
      \end{tabular}
    }
    \caption{Subsequent filled $T_1$, and $T_2$.}
    \label{example2:table:3}
\end{figure}

\begin{figure}[!ht]
    \centering
    \resizebox{\columnwidth}{!}
    {
      \begin{tabular}[t]{c|c}
        \imagetop{\includegraphics[page=7]{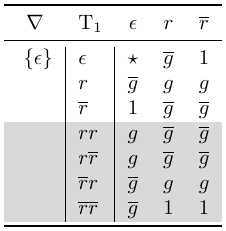}} &
        \imagetop{\includegraphics[page=15]{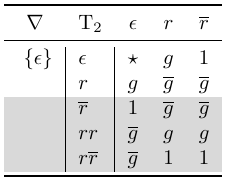}}
      \end{tabular}
    }
    \caption{$T_1$, and $T_2$ after closing. Since neither of them cover row $rr$, both are not yet closed.}
    \label{example2:table:5}
\end{figure}

\begin{figure}[!ht]
    \centering
    \resizebox{\columnwidth}{!}
    {
      \begin{tabular}[t]{c|c}
        \imagetop{\includegraphics[page=9]{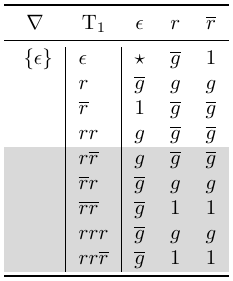}} &
        \imagetop{\includegraphics[page=17]{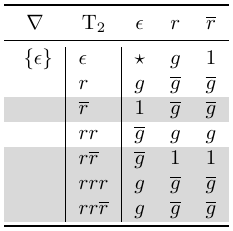}}
      \end{tabular}
    }
    \caption{Closed $T_1$, and $T_2$.}
    \label{example2:table:7}
\end{figure}

\begin{figure}[!ht]
    \centering
    \resizebox{\columnwidth}{!}
    {
      \begin{tabular}[t]{c|c}
        \imagetop{\includegraphics[page=10]{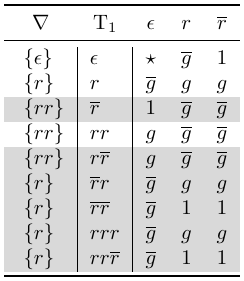}} &
        \imagetop{\includegraphics[page=18]{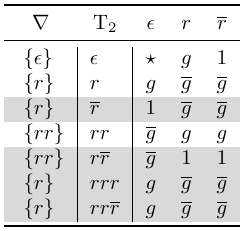}}
      \end{tabular}
    }
    \caption{Final $T_1$, and $T_2$ after computing the cover-set $\nabla$ and minimizing $B$. 
    From these tables we can extract the symbolic transducers depicted in Fig.~\ref{example2:hyp:0}.}
    \label{example2:table:8}
\end{figure}

\subsection{Example 3}
Last we consider the example given bellow Corollary~\ref{cor:termination}. This is an example
showing the algorithm converges in spite of the fact that the language contains a non-regular
language, in the sense that its rank is not finite.

\begin{figure}[!ht]
    \centering
    \begin{tabular}[t]{c|c}
        \includegraphics[page=1]{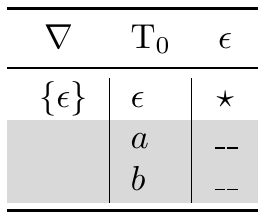} &
        \includegraphics[page=2]{examples/example3.pdf}
    \end{tabular}
    \caption{The initial table is on left, and on right it is filled.}
    \label{example3:table:0}
\end{figure}

\begin{figure}[!ht]
    \centering
    \includegraphics[page=4]{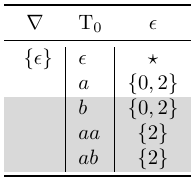}
    \caption{Table after closing, since rows $aa$ and $ab$ are not covered, the table is not closed yet.}
    \label{example3:table:1}
\end{figure}

\begin{figure}[!ht]
    \centering
    \includegraphics[page=6]{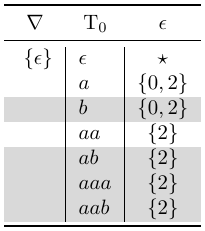}
    \caption{Subsequent closed table.}
    \label{example3:table:2}
\end{figure}

\begin{figure}[!ht]
    \centering
    \includegraphics[page=8]{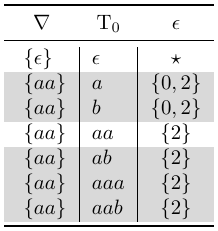}
    \caption{Subsequent rebased table. We removed $a$ from $B$ and $\nabla$ of all rows.}
    \label{example3:table:3}
\end{figure}

\begin{figure}[!ht]
    \centering
    \includegraphics[page=9]{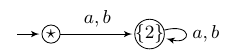}
    \caption{We extract the above conjecture transducer from the last table; the oracle provides no counterexamples for this transducer.}
    \label{example3:table:4}
\end{figure}

\newpage
\section{Omitted Implementation Details}\label{app:imp-details}

We implemented Algorithm~\ref{alg:sfour} and executed it on a number of examples, described in Section~\ref{sec:experiments-desc}. Table~\ref{tab:all-results} presents the results in terms of the
number of queries, the number of tables, the number of splits, the size of the resulting transducer, 
and the running time. In these experiments, the target bi-languages were generated from specifications
of arbiters. 
For most experiments, we implemented membership oracles through a logic-based method. 
We first transforms an expanded temporal formula into a disjunctive normal form (DNF), 
and then evaluated membership of a bi-word using Boolean falsifiability. 

As Table~\ref{tab:all-results} shows, the oracles consumes a significant amount of
time, often much more than the learning algorithm \SFour itself.
This is since Boolean falsifiability on DNF is NP-Hard.\footnote{
	To use presented method for synthesizing reactive systems from temporal specifications,
	we suggest to implement a membership oracles incorporating all known optimizations introduced in translating temporal formulas to $\omega$-automata.}

\subsection{Possible Improvements}
We suggest several improvements of the implementation of the algorithm.

First, it is easy to see that the proposed method allows parallel exploration of symbolic tables. 
Parallel calls to multiple oracle instances while filling table entries that are not related would also
lead to a performance improvement.

A further improvement can be achieved by adding more sophisticated heuristics for the
traversal of the tree of symbolic tables.

\subsection{Experiments Description}\label{sec:experiments-desc}
We provide the experiments in terms of \emph{Linear Temporal Logic} (LTL) formulas~\cite{Pnueli77}.
For the reader unfamiliar with LTL, we mention that 
the formula $\G \varphi$ (read \emph{globally} $\varphi$)
states that the formula $\varphi$ should hold on every cycle starting the current cycle,
the formula $\X \varphi$  (read \emph{next} $\varphi$) states that $\varphi$ should
hold on the next cycle, and the formula $\X^n \varphi$ abbreviates $\underbrace{\X \X \ldots \X}_{\text{$n$ times}} \varphi$.

\subsection*{Experiment 1}
Each request is granted either in the current step or the subsequent step. Please note that this is the initial illustrative example.
\[
\G(r \to (g \lor \X g))\;.
\]

\subsection*{Experiment 2} 
Each request is granted either in the current step or the subsequent step; meanwhile, grant is always lowered subsequently. Please note that this is the secondary illustrative example.
\[
\G(r \to (g \lor \X g)) \land \G(g \to \X \overline{g})\;.
\]

\subsection*{Experiment 3} 
Every request from the fourth step onward is granted in the subsequent step; that is,
\[
\X^4\G(r \to \X g)\;.
\]

\subsection*{Experiment 4} 
This experiments shows it is possible to have alphabets other than Boolean literals.
We fix the system's interface to $\Sigma=\{a\}$, $\Gamma=\{1,2\}$. The target hidden specification is
the first three inputs are immediately responded with $1$; thereupon, with $2$.

\subsection*{Experiment 5} 
Each request is granted in the current step; meanwhile, the output sequence $gg\overline{g}$ is forbidden.
\[
\G( r \to g ) \land \G\neg(g \land \X g \land \X^2 \overline{g})
\]

\subsection*{Experiment 6} 
The target hidden language is the irregular one we described bellow Corollary~\ref{cor:termination}.

\subsection*{Experiment 7} 
The target specification is an arbiter of $n$ clients.
Each client's request is granted latest at $n-1$ steps after its arrival:
\[
\bigwedge_{i\leq n } \G (r_i \to \bigvee_{t < n } \X^t g_i )\;,
\]
and grants are always mutually exclusive:
\[
\bigwedge_{i\leq n }\bigwedge_{i<j\leq n }\G( \overline{g}_i \lor \overline{g}_j )\;.
\]

\subsection*{Experiment 8} 
The target specification is an arbiter of $n$ clients.
If a client's request goes down then, grant follows in the next step:
\[
\bigwedge_{i\leq n } \G (\overline{r}_i \land g_i \to \X \overline{g}_i )\;,
\]
also if exists an open request then, any open request should be granted in the next step:
\[
\G ( \bigvee_{i\leq n } r_i  \to \bigvee_{i\leq n } (r_i \land \X g_i) )\;,
\]
and finally, grants are always mutually exclusive:
\[
\bigwedge_{i\leq n }\bigwedge_{i<j\leq n }\G( \overline{g}_i \lor \overline{g}_j )\;.
\]

\subsection*{Experiment 9} 
The target language of an arbiter whose $1$\textsuperscript{st} and $n$\textsuperscript{th} outputs are indirectly related as follows:
\[
(g \to \X^{n} g)\,;
\]
meanwhile, as of the $1$\textsuperscript{st} step up to $n$\textsuperscript{th} step (excluding them) grant immediately agrees with request, that is:
\[
 \bigwedge_{0 < i < n } \X^i(r \leftrightarrow g)\;,
\]
also, the $n$\textsuperscript{th} output is determined by the value of arbiter's previous output; that is,
\[
\X^{n-1}g \leftrightarrow \X^{n} g\,.
\]
Finally, as of $n$\textsuperscript{th} step (excluding that step), grants agree with requests with a delay of $n$ steps; that is,
\[
 \X^{n+1}\G(r \leftrightarrow \X^ng)\;.
\]

\subsection*{Experiment 10} 
The target language of a parameterized arbiter that implements a combination lock using Fibonacci series.
That is, given a parameter $n$, following is the initialization sequence of the arbiter:
\[
\bigwedge_{0<i<n} \X^{\operatorname{fib}(i+2)} r
\]
if initialized correctly, grant immediately agrees with the request, that is
\[
 (\bigwedge_{0<i<n} \X^{\operatorname{fib}(i+2)} r) \to \X^{\operatorname{fib}(n+1)} \G(r \leftrightarrow g)
\]
otherwise, grant is never raised, that is
\[
(\neg \bigwedge_{0<i<n} \X^{\operatorname{fib}(i+2)} r) \to \G(\overline{g})\;.
\]

\end{document}

%% file: submission/preambles.tex
\usepackage{times}
\usepackage{helvet}
\usepackage{courier}
\usepackage[hyphens]{url} % DO NOT CHANGE THIS
\usepackage{graphicx} % DO NOT CHANGE THIS
\usepackage{graphicx}  % DO NOT CHANGE THIS
\usepackage{caption}

%% file: preamble/packages.tex
\usepackage[utf8]{inputenc}
\usepackage[english]{babel}
\usepackage{amssymb, amsmath, amsfonts, amsthm}
\usepackage{mathtools}

\usepackage{calc} %for widthof
\usepackage{graphicx}
\usepackage{acronym}
\usepackage{mathtools} %for smallmatrix used in biword
\usepackage{stmaryrd} %for \sema
\usepackage{mathrsfs}
\usepackage{xcolor}
\usepackage{csvsimple}
\usepackage{array,booktabs}
\usepackage{multirow}
\usepackage{textcomp}
\usepackage{xspace}

\input{preamble/algorithms.tex}

\input{preamble/tikz.tex}

%% file: preamble/algorithms.tex
\usepackage{algorithm}
\usepackage[noend]{algpseudocode}

\algrenewcommand\alglinenumber[1]{{\footnotesize#1}}
\algrenewcommand\algorithmicrequire{\textbf{Precondition:}}
\algrenewcommand\algorithmicensure{\textbf{Postcondition:}}
\algnewcommand\Not[1]{\textbf{not} #1}
\algnewcommand\Continue{\textbf{continue}}
\algnewcommand\Break{\textbf{break}}

\newcommand*\Let[2]{\State #1 $\gets$ #2}

%% file: preamble/tikz.tex
\usepackage{tikz}
\usetikzlibrary{arrows,shapes,automata,positioning,decorations,calc}
\tikzstyle{state}=[draw,circle,minimum size=0.5cm]
\tikzstyle{automaton}=[>=latex',line join=bevel,initial text={},node distance=1.5cm,shorten >=1pt,every path/.style={->}]

%% file: preamble/definitions.tex
\newcommand\mepast[1]{}
\newcommand\hrpast[1]{}
\newcommand\hcpast[1]{}
\newcommand\dfpast[1]{}
\newcommand\commentout[1]{}

\newcommand{\lstar}{L\textsuperscript{*}}
\newcommand{\SFour}{S\textsuperscript{4}\xspace}

\newcommand{\st}{s.\,t.\ }

\newcommand{\newbase}[1]{\ensuremath{\nabla'(#1)}}
\newcommand{\prefixes}[1]{\ensuremath{\text{prefixes}(#1)}}
\newcommand{\suffixes}[1]{\ensuremath{\text{suffixes}(#1)}}
\newcommand{\set}[1]{\ensuremath{\mathbb{#1}}}

\newcommand\X{\operatorname{\mathbf{X}}}

\newcommand\G{\operatorname{\mathbf{G}}}

\def\ialphit{\ensuremath{\mathit{\Sigma}}}
\def\oalphit{\ensuremath{\mathit{\Gamma}}}

\def\ialph{\ensuremath{\mathrm{\Sigma}}}
\def\oalph{\ensuremath{\mathrm{\Gamma}}}
\def\sym{\ensuremath{x}}
\def\isym{\ensuremath{\sigma}}
\def\osym{\ensuremath{\gamma}}
\def\w{\ensuremath{w}}
\def\iw{\ensuremath{v}}
\def\ow{\ensuremath{w}}
\def\qm{\texttt{?}}

\newcommand{\claimref}[1]{\textbf{#1}} 

\newcommand{\project}[2]{\ensuremath{{#1}_{\mathbin{\downharpoonright}{#2}}}}
\newcommand{\given}[2]{\ensuremath{{#1}_{\mathbin{\vert}{#2}}}}
\newcommand{\projectL}[1]{\project{L}{#1}}
\newcommand{\givenL}[1]{\given{L}{#1}}
\newcommand{\tab}[1]{\ensuremath{\mathrm{#1}}}

\newcommand{\emptysymb}{\ensuremath{\_\_}}

\newtheorem{theorem}{Theorem}[section]
\newtheorem{proposition}[theorem]{Proposition}
\newtheorem{corollary}[theorem]{Corrolary} 
\newtheorem{lemma}[theorem]{Lemma} 
\newtheorem{claim}[theorem]{Claim}

\def\imagetop#1{\vtop{\null\hbox{#1}}}

%% file: preamble/acronyms.tex
\newacro{DFA}{deterministic finite automaton}
\newacro{MAT}{minimally adequate teacher}
\newacro{SUL}{system under learning}

\acrodefplural{DFA}{deterministic finite automata}
\acrodefplural{SUL}{systems under learning}

\newacro{LTL}{Linear Temporal Logic}
\newacro{OGIS}{Oracle Guided Inductive Synthesis}
\newacro{S3}{S\textsuperscript{3}}

\newcommand{\sema}[1]{\ensuremath{{\llbracket}#1{\rrbracket}}}
\newcommand{\semainf}[1]{\ensuremath{{\llbracket}#1{\rrbracket}_\omega}}
\newcommand{\semafin}[1]{\ensuremath{{\llbracket}#1{\rrbracket}_*}}
\newcommand{\biword}[2]{\ensuremath{\left[\begin{smallmatrix} #1 \\ #2 \end{smallmatrix}\right]}}
\newcommand{\tuple}[1]{\ensuremath{\langle #1 \rangle}}
\newcommand{\tree}[1]{\ensuremath{\mathcal{#1}}}
\newcommand{\aut}[1]{\ensuremath{\mathcal{#1}}}
\newcommand{\concTrees}{\ensuremath{\textsl{concTrees}}}
\newcommand{\symbTrees}{\ensuremath{\textsl{symbTrees}}}

\newcommand{\bools}{\ensuremath{\mathcal{B}}}
\newcommand{\specialsymbol}{\star}

\newcommand{\sytb}[1]{\mathcal{#1}}

\newcommand{\true}{\ensuremath{\textsl{true}}}
\newcommand{\false}{\ensuremath{\textsl{false}}}

\newcommand{\smq}{\ensuremath{\textsc{smq}}}

\newcommand{\scq}{\ensuremath{\textsc{scq}}}

\newcommand{\eq}{\ensuremath{\textsc{eq}}}
\newcommand{\mq}{\ensuremath{\textsc{mq}}}

\newcommand{\method}[1]{\ensuremath{\textsl{#1}}}
\newcommand{\isClosed}{\method{isClosed}}
\newcommand{\extractAut}{\method{extractTransducer}}

\newcommand{\fillTb}{\method{fill}}
\newcommand{\Split}{\method{split}}
\newcommand{\rebase}{\method{rebase}}
\newcommand{\findShortestCE}{\method{findShortestCE}}

\newcommand{\incompatible}[1]{\ensuremath{\textsl{incompatibility}}(#1)}
\newcommand{\rank}[1]{\ensuremath{\textsl{rank}}(#1)}

%% file: submission/authors.tex
%\author{Anonymized}
\author{
    Roderick Bloem\\
    \small TU Graz\\
\and
    Hana Chockler\\
    \small King's College London\\
\and
    Masoud Ebrahimi\\
    \small TU Graz\\
\and
    Dana Fisman\\
    \small Ben-Gurion University\\
\and
    Heinz Riener\\
    \small EPFL
}

%% file: submission/abstract.tex
\commentout{
We examine the problem of synthesis of reactive systems, when the specification is unknown. 
We propose to replace synthesis by learning a system from a set of example traces.
The existing body of results regarding learning regular $\omega$-languages essentially provides a way to learn a specification of a reactive system, rather than
a system itself. It thus does not provide an alternative for synthesis. The challenge in learning a system compared to learning a specification is that 
good examples can conform to different implementations. A learning algorithm should be able to cope with examples coming from possibly conflicting implementations and
still synthesize a system that meets the unknown specification.
We present a theoretical framework for synthesizing safety specifications in such a setting and its implementation. 
%We discuss possible applications of our approach.
}

We define the problem of learning a transducer $\aut{S}$ from a target language $U$ containing possibly conflicting transducers, using membership queries and conjecture queries.
The requirement is that the language of $\aut{S}$ be a subset of $U$.
We argue that this is a natural question in many situations in hardware and software verification.
We devise a learning algorithm for this problem and show that its time and query complexity is polynomial with respect to the rank of the target language, its incompatibility measure, and
the maximal length of a given counterexample. We report on experiments conducted with a prototype implementation. 

%% file: sections/introduction.tex
\section{Introduction}

Constructing reliable systems is a main requirement and a major challenge in
safety critical systems such as autonomous vehicles, medical devices and banking systems.
Formal verification methods can be used to find bugs or increase assurance that the system satisfies
its requirements. 
A leading alternative approach is to automatically synthesize a 
correct-by-construction system from a given formal specification of its requirements.
This line of research, termed \emph{system/program synthesis}~\cite{MannaW80,PnueliR88}, assumes the existence of a
perfect specification that fully characterizes the set of correct behaviors. Such a specification is often
as hard to write as the system itself~\cite{Kupferman16,McM19}. %Synthesizing imperfect specifications result in incorrect systems.
Other criticisms over the setting assumed by system synthesis are that it considers the case where the system is designed from scratch
rather than from a  previous version, or by using library components~\cite{LustigV13}. 

For this reason, the synthesis problem, is now taking relaxed forms, which assume a \emph{specification scale} which on
one end has complete rigorous specifications such as temporal logics, somewhere in the middle it has partial implementations and/or  partial specifications, 
and on the other far end it has merely examples~\cite{Gul11,NarodytskaLBRW14,DSH17,ASFS18,NeiderG18}. The line of research in this end is often termed \emph{example-driven programming}.

There are various ways to define the problem of synthesizing systems from examples. 
Many such examples can be found in the surveys by~\cite{DBLP:journals/cacm/Vaandrager17,Fisman18}.
Most extend the \lstar\ algorithm of~\cite{Angluin87} that learns regular languages using \emph{membership queries} and \emph{equivalence queries}
to learning different types of automata, e.g. multiplicity automata~\cite{BV96},
Mealy machines~\cite{Niese03,DBLP:conf/fm/ShahbazG09}, I/O-Automata~\cite{AartsV10}, weighted automata~\cite{BM15},
symbolic automata~\cite{MM16,DrewsD17} and more. 

We are interested in \emph{reactive systems}, systems that interacts with their environment on an ongoing basis.
In \emph{formal verification}  such systems are modeled using
languages of infinite words that represent the ongoing nature of the system. Learning of languages of infinite words,
is also a well studied subject~\cite{FarzanCCTW08,MalerP95,AngluinF14,AngluinAF20}. 

In \emph{formal synthesis} reactive systems are implemented by transducers. A \emph{transducer} is a finite state machines where each transition is  labeled by an input symbols (from a set $\Sigma$),
and each state is  labeled by an output (from a set  $\Gamma$). If on reading input word $\isym_1 \isym_2 \isym_3,\ldots$ the transducer $\aut{S}$ visits state $q_1,q_2,q_3,\ldots$ and state $q_i$ is labeled by $\osym_i$,
then the prefixes of  $\biword{\isym_1}{\osym_1}\biword{\isym_2}{\osym_2}\biword{\isym_3}{\osym_3}\dots$ 
are in the language of $\aut{S}$, denoted $\sema{\aut{S}}_*$. The language of $\aut{S}$ thus consists of words over $\Sigma\times\Gamma$, but not any subset of $(\Sigma\times\Gamma)^*$ corresponds to a transducer. The language of a transducer
is regular; \emph{exhaustive}, meaning when projected on $\Sigma$ it consists of all words in $\Sigma^*$; and satisfies the property that every input word is matched with a single output word. The literature on learning transducers (or Mealy machines) assumed the target language adheres to these requirements.

The last requirement entails that the assumption is that the target language conforms to a single implementation. We challenge this assumption. Formally, we are interested in the following problem: devise a learning algorithm using membership queries and conjecture queries, that learns an unknown regular exhaustive language $U$ over $\Sigma\times\Gamma$, even if the examples may correspond to different implementations. The requirement is to output a transducer $\aut{S}$ such that $\sema{\aut{S}}_*\subseteq \prefixes{U}$.

We argue that this is a natural question in various settings. Verification of software and hardware systems usually assumes a model of the environment the system under verification interacts with.
In cases where the environment can be any of a number of third-party black box systems, modeling the environment is a challenging task.
For instance, this is the case with the Amazon Prime Video app, which should work on all mobile phones, televisions, laptops, and tablet devices. 
The verification team of Prime Video might collect executions from different
platforms with only little ability to traceback the platform, version etc. and then attempt to build models of these platforms, in order to verify the app. 
Another scenario is a heterogeneous framework, consisting of software interacting with humans and
with third-party robots~\cite{WWADEFP20}. In such cases we would like to obtain 
a transducer that encompasses as many behaviors as possible of the black-box components of the environment.

The main challenge we face, compared to other literature on learning automata, is the fact that 
the target language is only assumed to contain the language of a desired transducer, but it may in fact
consist of many other words, in particular  words corresponding to different implementations.

To understand what we mean by conflicting implementations, consider for instance the specification ``always if $o$ holds, then $o$ does not hold in the next
cycle'' (where $o$ is an output signal). In this case the learner can get answers that correspond to outputting $o$ on every even tick, as well as answers that correspond to outputting $o$ at every odd tick, and trying to build a transducer that adheres to both would lead her astray.

In learning Mealy or Moore machines, the natural generalization of \lstar's membership query is a query that takes an input sequence $\iw=\isym_1\isym_2\ldots \isym_k$ and returns the output $\gamma_k$ the transducer emits on reading $\iw$. Incorporating this style of membership queries to our setting is problematic. Consider for instance, the case where $\Sigma=\{i\}$,  $\Gamma=\{o,\overline{o}\}$ and the specification says all sequences are allowed. On membership queries for input sequences in $i^*$,  the oracle could provide answers that are consistent with the  following input-output trace 
\[\biword{i}{o} \cdot \biword{i}{\overline{o}}^1 \cdot  \biword{i}{o} \cdot \biword{i}{\overline{o}}^2 \cdot  \biword{i}{o} \cdot \biword{i}{\overline{o}}^3 \cdot  \biword{i}{o} \cdot \biword{i}{\overline{o}}^4 \cdot  \biword{i}{o}  \cdots\]
that satisfies the specification, but has no finite state machine realizing it. 

Another obstacle can be illustrated by considering a specification such as ``every request should eventually be granted''. On any sequence with a request followed by $n$ cycles with no grant, the answer to the membership query should be ``yes''. In other words, every finite sequence is allowed, though clearly not every infinite sequence is allowed.
For this reason we focus on \emph{safety languages}, that is, those for which every counterexample has a finite witnessing prefix. Working with safety languages also solves the issue, that in the scenarios where examples come from black-box implementations, there is no way to obtain an infinite behavior.

%Our algorithm, roughly speaking,  assumes an oracle answering two types of queries: \emph{conjecture queries} that answer whether the conjecture implementation is correct (and providing a counterexample when this is not the case), and \emph{membership %queries} that approve/disapprove a possible behavior.
 To cope with the fact that we may get answers for different implementations, we work with symbolic transducers (and in accordance symbolic conjecture queries and symbolic membership queries). A \emph{symbolic transducer} is a state machine in which the transition between states corresponds to the input read, and each state is labeled by a set of outputs (that may be emitted on words leading to that state). As in other Angluin-style algorithms, we use a data structure termed an observation table, where we keep answers to the membership queries, and we try to distinguish states of the desired transducer. Since we may be dealing with several conflicting implementations, it may not be possible to keep the information in one table from which a transducer can be extracted. In such cases our algorithm splits into several tables, keeping track of different implementations (where one table can track several implementations, as long as they are compatible as we formally explain later). 
 We analyze the complexity of the algorithm with respect to  two measures we define on exhaustive languages, the \emph{rank} and \emph{incompatibility} measure. We show that the algorithm is polynomial in these measures (as well as the maximal length of a received counterexample).
 
\vspace{-2mm}
\paragraph{An illustrative example}
Consider the unknown language $U$ consisting of behaviors that grant requests either in the step where the  request was received or in the next step. The input variable is $r$ (request), and the single output
variable is $g$ (grant). There exist an infinite number of concrete transducers realizing $U$,
as, for example, for any given $k$, we can construct a transducer that grants at the same step for the first
$k$ steps and from the step ${k+1}$ onwards it grants in the step after the request. However, the concrete
transducers are not conflicting, that is, they can be represented by a single 
symbolic transducer, e.g. the transducer  $\aut{S}_1$ depicted in Fig.~\ref{example1:hyp:0}. 
\begin{figure}
    \centering
    \begin{tikzpicture}[automaton]
      \small
      \node[state, minimum size=0.65cm] (1) at(0,0) {};
      \node[state, right= of 1] (3) {};
      \node[initial, draw, circle, yshift=1cm] (0) at ($(1)!0.5!(3)$) {};
      \node at (0) {$\star$};
      \node at (1) {$g,\overline{g}$};
      \node at (3) {$g$};
      \path
        (0) edge [above left=10] node [left, near start] {$r,\overline{r}$} (1)
        (0) edge [above right] node [right, near start ] {$\overline{r}$} (3)
        (1) edge [bend left=10] node [auto] {$r,\overline{r}$} (3)
        (3) edge [bend left=10] node [auto] {$r,\overline{r}$} (1);
    \end{tikzpicture} 
%    \begin{tikzpicture}[automaton]
%      \small
%      \node[initial, draw, circle] (0) at (0,0) {};
%      \node[state, right=1 of 0] (3) {};
%      \node at (0) {$\star$};
%      \node at (3) {$g$};
%      \path
%        (0) edge [above right] node [auto] {$r, \overline{r}$} (3)
%        (3) edge [loop right] node [auto] {$r,\overline{r}$} (3);
%    \end{tikzpicture} 
    \caption{A symbolic transducer $\aut{S}_1$.}
    \label{example1:hyp:0}
\end{figure}

Consider now the language that, in addition to the previous description, disallows two subsequent grants. In this case, there exist two
conflicting implementations realizing the language: one that grants in the current step,
and another that grants in
the next step (see Fig.~\ref{example2:hyp:0}). Because of the additional constraint, these
implementations cannot coexist, i.e. cannot be modeled by the same symbolic transducer. 
In such cases, our algorithm outputs a symbolic transducer that
represents some non-conflicting implementations. 
\begin{figure}
    \centering
    %\includegraphics[page=11]{examples/example2.pdf}
% \begin{tabular}{ll}
%     \begin{tikzpicture}[automaton]
%   \small
%   \node[state] (1) at(0,0) {};
%   \node[state, right= of 1] (3) {};
%   \node[initial, draw, circle, yshift=1cm] (0) at ($(1)!0.5!(3)$) {};
%   \node at (0) {$\specialsymbol$};
%   \node at (1) {$\overline{g}$};
%   \node at (3) {$g$};
%   \path
%     (0) edge node [above left] {$r$} (1)
%     (0) edge node [above right] {$\overline{r}$} (3)
%     (1) edge [bend left=10] node [auto] {$r,\overline{r}$} (3)
%     (3) edge [bend left=10] node [auto] {$r,\overline{r}$} (1);
%  \end{tikzpicture} 
%  &
%   %  \includegraphics[page=19]{examples/example2.pdf}
% \begin{tikzpicture}[automaton]
%   \small
%   \node[state] (1) at(0,0) {};
%   \node[state, right= of 1] (3) {};
%   \node[initial, draw, circle, yshift=1cm] (0) at ($(1)!0.5!(3)$) {};
%   \node at (0) {$\specialsymbol$};
%   \node at (1) {$g$};
%   \node at (3) {$\overline{g}$};
%   \path
%     (0) edge node [ left, near start] {$r,\overline{r}$} (1)
%     (1) edge [bend left=10] node [auto] {$r,\overline{r}$} (3)
%     (3) edge [bend left=10] node [auto] {$r,\overline{r}$} (1);
% \end{tikzpicture} 
% \end{tabular}
\begin{tabular}{ll}
    \begin{tikzpicture}[automaton]
  \small
  \node[state] (1) at(0,0) {};
  \node[state, right= of 1] (3) {};
  \node[initial, draw, circle, yshift=1cm] (0) at ($(1)!0.5!(3)$) {};
  \node at (0) {$\specialsymbol$};
  \node at (1) {$\overline{g}$};
  \node at (3) {$g$};
  \path
    (0) edge node [left, near start] {$r, \overline{r}$} (1)
    (0) edge node [auto] {$\overline{r}$} (3)
    (1) edge [bend left=10] node [auto] {$r,\overline{r}$} (3)
    (3) edge [bend left=10] node [auto] {$r,\overline{r}$} (1);
 \end{tikzpicture} 
 &
\begin{tikzpicture}[automaton]
  \small
  \node[state] (1) at(0,0) {};
  \node[state, right= of 1] (3) {};
  \node[initial, draw, circle, yshift=1cm] (0) at ($(1)!0.5!(3)$) {};
  \node at (0) {$\specialsymbol$};
  \node at (1) {$g$};
  \node at (3) {$\overline{g}$};
  \path
    (0) edge node [ left, near start] {$r,\overline{r}$} (1)
    (0) edge node [ auto ] {$\overline{r}$} (3)
    (1) edge [bend left=10] node [auto] {$r,\overline{r}$} (3)
    (3) edge [bend left=10] node [auto] {$r,\overline{r}$} (1);
\end{tikzpicture} 
\end{tabular}
    \caption{Two symbolic transducers $\aut{S}_2$ (on the left) and $\aut{S}_3$ (on the right) that cannot coexist.}
    \label{example2:hyp:0}
    \vspace{-2mm}
\end{figure}

\vspace{-3mm}
\paragraph{Outline}
We provide definitions and notations in Sec.~\ref{sec:defs}, the learning algorithm in Sec.~\ref{sec:learn}
and its correctness proof and complexity analysis in Sec.~\ref{sec:correctness}. Experimental results are
given in Sec.~\ref{sec:experminetal} and we conclude with a discussion in Sec.~\ref{sec:discuss}.
The reader is referred to the supplementary material for the full details of the illustrative examples,
as well as a number of more complex examples. 
Due to the lack of space, some proofs are also moved to the supplementary
material.
%All transducers in the paper and in the supplementary material
%were produced using our prototype implementation of the algorithm.

%% file: sections/formalization.tex
%\subsection*{Notations}
\section{Definitions and Notations}\label{sec:defs}
We make use of the following notations and definitions.
An alphabet $\Sigma$ is a non-empty finite set of symbols. 
The set of all finite words over $\Sigma$ is denoted $\Sigma^*$, 
the set of all $\omega$-words (infinite words) over $\Sigma$ is denoted $\Sigma^\omega$, and the set of finite and infinite words is denoted $\Sigma^\infty$. Words are indexed starting $1$. 
That is, $\w = \sym_1 \sym_2 \dots$. 
The length of a finite word $\w = \sym_1 \sym_2 \dots \sym_m$, denoted $|\w|$, is $m$. 
Given $\w = \sym_1 \sym_2 \dots$ 
the  $i$-th letter of $\w$ is denoted $\w[i]$,
the prefix of $w$ ending in $\w[i]$ is denoted $\w[..i]$,  
the suffix of $\w$ starting at $\w[i]$ is denoted $\w[i..]$,  
and the infix of $\w$ starting at $\w[i]$ and ending in $\w[j]$, is denoted $\w[i..j]$.
A finite word $\w$ is said to be a prefix of a finite/infinite word $\w'$, 
denoted $\w\preceq \w'$, if there exists $j<|\w'|$ such that $\w=\w'[..j]$. 
We use $\prefixes{L}$ for the set of prefixes of words in $L$. Similarly,
we use $\suffixes{L}$ for the set of suffixes of words in $L$.
%We use $\replaceindinword{w}{i}{\sym}$ for the word obtained from $w$ by changing the $i$-th letter to $\sym$.

\subsection*{Regular Trees}
A \ialphit-tree $T$ is a non-empty prefix closed subset of $\ialph^*$. 
We think of \ialphit\ as the directions of the tree.
We view $\epsilon$ as the root of the tree, and for every $\w\isym\in T$ we view the word $\w\isym$ as the child of $\w$ in direction $\isym$. 
A \oalphit-labeled \ialphit-tree $\tree{T}$ is a pair $\tuple{T,\tau}$ such that $T$ is a \ialphit-tree and $\tau:T\rightarrow \oalph$ maps every node of the tree $T$ to a label in $\oalphit$. 
A \oalphit-labeled \ialphit-tree $\tuple{T,\tau }$ is said to be \emph{exhaustive} if $T=\ialph^*$. 
Let $\tree{T}=\tuple{\ialph^*,\tau}$ be an exhaustive  \oalphit-labeled \ialphit-tree. 
A word $\w \in \ialph^*$ induces a sub-tree $\tree{T}_\w=\tuple{\ialph^*,\tau_\w}$, also an exhaustive \oalphit-labeled \ialphit-tree, where $\tau_\w(\iw)=\tau(\w\iw)$ for every $\iw\in\ialph^*$.  
An exhaustive labeled tree $\tree{T}$ is said to be \emph{regular} if it contains a finite number of non-isomorphic sub-trees.

%\subsection*{Realizable Bi-Languages,\\Contained Symbolic/Concrete Tree}
\subsection*{Realizable Bi-Languages, Contained Trees}
Let $\ialph$ and $\oalph$ be two alphabets. A word over $\ialph\times\oalph$ is referred to as a bi-word. A language over $\ialph\times\oalph$ is referred to as a bi-language.
Let $v=\isym_1\isym_2\isym_3\ldots\in\ialph^\infty$, $w=\osym_1\osym_2\osym_3\ldots\in\oalph^\infty$ be two words of equal length.
We use $v \oplus w$ for the bi-word  $\biword{\isym_1}{\osym_1}\biword{\isym_2}{\osym_2}\biword{\isym_3}{\osym_3}\ldots$ over $\ialph \times \oalph$.
Given a bi-language $L$, we use \projectL{\ialph} to denote the projection of $L$ on $\ialph$, namely the set of words $\{ v \in \ialph^\infty ~|~ \exists w \in \oalph^\infty \text{ \st } v \oplus w \in L\}$. 
A bi-language $L\subseteq(\ialph\times\oalph)^\omega$ is said to be \emph{\ialphit-exhaustive} if $\projectL{\ialph}=\ialph^\infty$. 
Given $\iw\in\ialph^\infty$ we use \given{L}{v} for the set of words $\{\ow\in\oalph^\infty ~|~ \iw \oplus \ow \in L \}$.

Henceforth, when we discuss bi-languages we assume they are over $\ialph\times\oalph$. Furthermore, we consider only exhaustive bi-languages. 
This is since we are interested in machines that provide answers to every possible sequence of inputs. Indeed, the language of a transducer (as formally defined in the sequel)
is always exhaustive. 
We refer to  exhaustive $\oalph$-labeled $\ialph$-trees as \emph{concrete trees}, and to  exhaustive
$2^\oalph$-labeled $\ialph$-trees as \emph{symbolic trees}.\footnote{We may represent a symbolic tree as a  $\bools(\oalph)$-labeled $\ialph$-trees, where $\bools(\oalph)$ is the set of Boolean expressions over $\oalph$, and  a  label $b$ is interpreted as the subset of letters in $2^\oalph$ satisfying $b$.}  Since the structure of an exhaustive tree is $\ialph^*$ by definition, we omit it from the description $\tree{T}=\tuple{\ialph^*, \tau}$, and identify $\tree{T}$ with $\tau$.

Let $L$ be a bi-language as above. We say that $L$ \emph{contains a concrete tree} $\tree{T}_C$ if for every $\iw\in\ialph^\omega$ there exists $\ow \in \givenL{v}$ such that $\tree{T}_C(u)  = \ow[|u|]$ for every $u\preceq \iw$. 
We use $\concTrees(L)$ to denote the set of concrete trees contained in $L$. 
We say that $L$ \emph{contains a symbolic tree} $\tree{T}_S$ if for every $v\in\ialph^\omega$ there exists $\alpha\in (2^\oalph)^\omega$ such that $\tree{T}_S(u)  = \alpha[|u|]$ for every $u\preceq v$ and $L_{|v}\supseteq \{w\in\oalph^\omega~|~\forall i \in \mathbb{N}.\ w[i]\in \alpha[i]\}$. We use $\symbTrees(L)$ to denote the set of symbolic trees contained in $L$. We say that $L$ is \emph{realizable} if there exists a regular-tree in $\concTrees(L)$. 
\mepast{the notations used in above paragraph are cryptic to my eyes. Maybe I am too stupid to know them by heart.
Do they have well-established connotations? Would someone please introduce us?}\dfpast{added a paragraph notation at the beginning. let me know if this is still unclear.}

It is not hard to see that being exhaustive is not a sufficient condition for realizability.

\begin{claim}\label{clm:exhaustive-not-realizable}
	$L$ may be \ialphit-exhaustive yet $\concTrees(L)= \emptyset$.
\end{claim}

The algorithm that we present in Section~\ref{sec:learn} works with symbolic trees. The following two claims assert that one can indeed search for a contained symbolic-tree
and extract from it a concrete tree, if so desired.

\begin{claim}\label{clm:conc-imp-symb}
	If $\concTrees(L)\neq \emptyset$ then $\symbTrees(L)\neq \emptyset$.
\end{claim}
\begin{claim}\label{clm:conc-in-symb}
	Let $\tree{T}_S$ be a symbolic tree in $\symbTrees(L)$. Let $\tree{T}_C$ be a concrete tree such that $\tree{T}_C(\iw)\in \tree{T}_S(\iw)$ for every $\iw\in\ialph^*$.
	Then $\tree{T}_C\in\concTrees(L)$.
\end{claim}

\subsection*{Symbolic, Concrete, and Consistent Transducers}
A symbolic transducer is a tuple $(\ialph,\oalph,Q,q_\iota,\delta,\eta)$ where $\ialph$ is the input alphabet, $\oalph$ is the output alphabet,
$Q$ is a finite non-empty set of states, $q_\iota$ is the initial state,  $\delta:Q\times\ialph \rightarrow 2^Q\setminus \emptyset$ maps
a state $q$ and an input letter $\isym$ to a non-empty set of possible next states,  
and $\eta:Q\rightarrow 2^\oalph\setminus \emptyset$ associates with each state $q$ a set of outputs that can be emitted when the transducer is in state $q$. Let  $w=\biword{i_1}{o_1}\biword{i_2}{o_2}\biword{i_3}{o_3}\ldots$ be a %n $\omega$-word. 
word. % RB: we want finite words to be generated as well.
We say that $w$ is \emph{generated} by $\aut{A}$ if there exists 
an infinite sequence of states
$q_0,q_1,q_2,\ldots $ such that $q_0=q_\iota$, $q_k \in \delta(q_{k-1},i_{k})$, a sequence of sets of output symbols $\theta_1,\theta_2,\ldots$ such that 
$\eta(q_k)=\theta_{k}$ and $o_k\in \theta_k$ for every $0<k\leq |w|$. We adopt the common extension of $\delta$ to work from $Q\times \ialph^*$ (to $2^Q\setminus\emptyset$).
The set of finite and infinite words \emph{generated} by a symbolic transducer $\aut{A}$, is denoted $\sema{\aut{A}}$.
We use $\semainf{\aut{A}}$ for $\sema{\aut{A}}\cap (\Sigma\times\Gamma)^\omega$ and $\semafin{\aut{A}}$ for $\sema{\aut{A}}\cap (\Sigma\times\Gamma)^*$.

\begin{claim}\label{clm:trans-fin-inf-safety} 
	$\prefixes{\semainf{\aut{A}}}=\semafin{\aut{A}}$ and $\semainf{\aut{A}}$ is a safety language.
\end{claim}

A symbolic transducer is \emph{concrete} if for every $q\in Q$ and $\isym\in\ialph$ we have $|\delta(q,\isym)|=1$ and $|\eta(q)|=1$.
A concrete transducer $\aut{C}=(\ialph,\oalph,q_\iota,Q,\delta,\eta)$ \emph{implements} the concrete tree  $\tree{T}_{\aut{C}}$ where $\tree{T}_{\aut{C}}(\epsilon)=\specialsymbol$, and for every $v\in\ialph^*$ such that $\delta(q_\iota,v)=q_v$ we have that $\tree{T}_{\aut{C}}(v)=\eta(q_v)$.\footnote{We use a special symbol $\specialsymbol$ for the output of the initial state, since we view an input-output sequence as starting with an input.}

\begin{claim}\label{clm:subtransducer}
Let $\aut{S}=(\ialph,\oalph,Q,q_\iota,\delta,\eta)$, and $\aut{S'}=(\ialph,\oalph,Q,q_\iota,\delta',\eta')$ be symbolic transducers s.t.
$\delta'(q,\isym)\subseteq \delta(q,\isym)$ and $\eta'(q)\subseteq \eta(q)$. Then $\sema{\aut{S}}\subseteq\sema{\aut{S}'}$.
\end{claim}
A symbolic transducer $\aut{S}$ is said to be \emph{consistent} if there exists a symbolic $2^\Gamma$-labeled $\Sigma$-tree $\tree{T}_\aut{S}$ such that 
the set of bi-words generated by the transducer $\aut{S}$ is exactly the set of words induced by the tree $\tree{T}_\aut{S}$. Note that a deterministic transducer
is always consistent, but a non-deterministic transducer may or may not be consistent.

\section{The Learning Algorithm}\label{sec:learn}
Before we provide the learning algorithm, we present its setting
and the data structures it uses.%, and their relation to concrete and symbolic transducers.

\subsection{Setting and Data Structures}
\subsubsection{Queries}
We consider two types of  queries: \emph{symbolic membership queries}, denoted \smq, and \emph{symbolic conjecture queries}, denoted \scq.
We may also use standard \emph{membership queries}, denoted \mq, which are derived from \smq s, as we explain in the following.

Let $U\subseteq(\ialph\times\oalph)^\omega$ be an unknown realizable bi-language.  The queries defined below are with respect to $U$.
\begin{itemize}
    \item A \emph{symbolic membership query} $\smq(\cdot)$ takes as input a finite non-empty word 
    $\biword{\isym_1}{\theta_1}
    \biword{\isym_2}{\theta_2}\ldots
    \biword{\isym_{m}}{\theta_{m}}
    \biword{\isym_{m+1}}{\qm}$ where $\isym_i\in\Sigma$ for every $1\leq i\leq m+1$ and $\theta_j\subseteq\Gamma$ for every $1\leq j \leq m$,  and returns the maximal subset $\theta_{m+1}$ of $\oalph$ such that the symbolic word obtained by replacing $\qm$ with $\theta_{m+1}$ is a subset of $\prefixes{U}$.
If the answer is $\emptyset$, it accompanies it with a concrete counterexample, namely a bi-word $\biword{\isym_1}{\osym_1}\biword{\isym_2}{\osym_2}\ldots\biword{\isym_m}{\osym_m}\biword{\isym_{m+1}}{\osym_{m+1}}\notin U$ for which
$\osym_i\in\theta_i$ for every $1\leq i \leq m$.

%Let $L$ be a symbolic language over $\ialph\times\oalph$ and let $L_*=L\cap (\ialph\times2^\oalph)^*$.
%A \emph{symbolic subset query} $\ssq(\cdot)$ takes as input a symbolic language $L$ and returns $\true$ if ${L}_* \subseteq \prefixes{U}$ and a word $u\in {L_*}\setminus \prefixes{U}$ otherwise. 

\item
A \emph{concrete membership query} $\mq$ takes a concrete finite bi-word $w$ over $\ialph\times\oalph$ and returns ``yes'' if $w\in U$ and ``no'' otherwise.
The query $\mq(\biword{\isym_1}{\osym_1}\biword{\isym_2}{\osym_2}\ldots\biword{\isym_{\ell}}{\osym_{\ell}})$ can be implemented using an \smq\ by checking
whether
\[\osym_\ell\in \smq(\biword{\isym_1}{\osym_1}\biword{\isym_2}{\osym_2}\cdots\biword{\isym_{\ell-1}}{\osym_{\ell-1}}  \biword{\isym_\ell}{\qm}).\]
%If the answer is $\theta_\ell$ then the result is ``yes'' iff $\osym_\ell\models\theta_\ell$.

\item
A \emph{symbolic conjecture query} $\scq(\cdot)$ takes as input a symbolic transducer $\aut{A}$ and returns $\true$ if $\semafin{\aut{A}}\subseteq \prefixes{U}$ and a word $u\in ({\semafin{\aut{A}}}\setminus \prefixes{U})$ otherwise.  
Note that it returns only \emph{negative counterexamples}. This is because we are looking for a transducer $\aut{S}$
such that $\sema{\aut{S}}\subseteq \prefixes{U}$. In particular, there may be many words in $\prefixes{U}$ 
that cannot be generated by $\aut{S}$. 
\end{itemize}

\subsubsection{Symbolic Observation Table}
Like $\lstar$, the algorithm makes use of a data structure called an \emph{observation table}. 
Unlike $\lstar$, the algorithm uses a \emph{symbolic} table, as defined next. 
%To be precise, it uses a set of tables.
A {symbolic table} over $\ialph,\oalph$ is a tuple $\sytb{H}=(R,C,M)$ where $R,C \subseteq \ialph^*$, $R$ and $C$ are prefix closed, and $M$ is an $|R|$ by $|C|$ matrix, where $M(r,c)$ is a subset of $\oalph$ 
or $\emptysymb$ (meaning that the entry has not yet been filled in). In addition, $M$ should satisfy that 
 for every $r,r'\in R$, $c,c'\in C$ if $rc=r'c'$ then $M(r,c)=M(r',c')$. The table is called \emph{filled} if it does not contain $\emptysymb$.
%Therefore we can use $M(w)$ for $w\in RC$ without specifying the decomposition of $w$ into two words $r\in R$ and $c\in C$ such that $w=rc$.
We use $M(r)$ to denote the sequence $M(r,c_1),\ldots, M(r,c_n)$, where $C=\{c_1,\ldots,c_n\}$. For two rows $r,r'$ we say that $M(r)$ \emph{implies} (or \emph{covers}) $M(r')$ if 
$M(r,c_i)\subseteq M(r',c_i)$ for every $c_i\in C$.

A filled symbolic observation table $\sytb{H}$ defines a set of bi-words, denoted  $\sema{\sytb{H}}$, defined as follows: $\sema{\sytb{H}}=\{{\iw\oplus \ow}\in(\ialph\times\oalph)^*~|~ \iw\in RC \mbox{ and }\forall {r\in R}, {c\in C}\mbox{ s.t. } r\cdot c\preceq \iw \mbox{ we have that } \ow[|rc|] \in M(r,c)\}$. We use $\sytb{H}(v)$ for the  set of words $\{v\oplus w~|~v\oplus w\in\sema{\sytb{H}}\}$.
The symbolic table $\sytb{H}$ is said to \emph{agree with a bi-language $L$} iff $\sema{\sytb{H}}\subseteq \prefixes{L}$.
%The symbolic table $\sytb{H}$ is said to \emph{agree with a symbolic transducer $\aut{A}$} iff $\sema{\sytb{H}}\subseteq \prefixes{\sema{\aut{A}}}$.
A symbolic transducer $\aut{A}$ is said to \emph{agree with a  symbolic table} $\sytb{H}=(R,C,M)$ iff for every $r\in R$ and $c\in C$ we have that ${\sema{\aut{A}}}(rc)\subseteq M(r,c)$.
A concrete finite tree $\tree{T}=\tuple{W,\tau}$ is said to be \emph{covered} by the table $\sytb{H}$ if $W\subseteq RC$ and for every $r\in R$ and $c\in C$ for which $rc\in W$ we have that $\tau(rc)\subseteq M(r,c)$.

A filled symbolic observation table $(R,C,M)$  is said to be \emph{closed} with respect to a subset $B\subseteq R$ termed a \emph{basis} if (i) $R$ and $C$ are prefix closed,
(ii) for every $b\in B$ and for every $\isym\in\ialph$ the word $b\isym$ is in $R$ and (iii) for every $b\in B$ and for every $\isym\in\ialph$ there exists a row $b'\in B$ such that  $M(b')$ implies $M(b\isym)$.{\dfpast{moved to demand a covered row $r\in R\setminus B$ to be equivalent (rather than implied) by a row $b\in B$.}}
We use $(R,C,M,B)$ for a symbolic table which is closed with respect to basis $B$. 
A closed table $(R,C,M,B)$ is said to be \emph{minimal} if $M(b)$ does not cover $M(b')$ for every $b,b'\in B$ s.t. $b\neq b'$.
Note that given a closed table $(R,C,M)$, the set of minimal elements in the partial order induced by implication of rows forms a minimal basis $B$.

%% RBNOTE: I think this should be moved to the end of Section 3.1
%% +
\subsubsection{Extracting a Transducer }
From a  closed  symbolic table $\sytb{H}=(R,C,M,B)$ we can extract a symbolic transducer 
$\aut{A}_{\sytb{H}}=(\ialph,\oalph,B,\epsilon,\delta,\eta)$ where for every $\isym\in\ialph$ and $b\in B$ we have that 
$\delta(b,\isym)=\{b'\in B ~|~ M(b') \text{ implies } M(b\isym)\}$ and $\eta(b)=M(b,\epsilon)$. Note that the resulting transducer
is \emph{non-deterministic} in general (since there may be a row $b\in B$ for which $r\sigma$ is implied by a set of basis rows $C=\{b_1,\ldots,b_k\}$, in which case there will be transitions from $b$ on $\sigma$ to the set $C$. %\{b_1,\ldots,b_k\}$). 

\begin{claim}\label{clm:extracted-trans}
    Let $\sytb{H}$ be a closed and minimal symbolic table, and $\aut{A}_{\sytb{H}}$ the transducer extracted from it.
    Then  $\aut{A}_{\sytb{H}}$ is a minimal transducers that agrees with $\sytb{H}$ and for any other minimal transducer $\aut{A}$ that agrees with $\sytb{H}$
    it holds that  $\sema{\aut{A}}\subseteq \sema{\aut{A}_{\sytb{H}}}$. 
\end{claim}

\subsection{The Learning Algorithm}\label{subsec:alg}
\input{sections/algorithms/sfour.tex}

The learning algorithm \SFour\ described in Alg.~\ref{alg:sfour} has access to oracles $\smq$ and $\scq$ (and $\smq$ which can be derived from $\smq$) that provide answers with respect to an unknown prefix-closed exhaustive bi-language $U$. 
It maintains a list of tables $\set{H}$ starting with a single table $\sytb{H}_0=(R,C,M,B)$  initialized 
with $R=\{\epsilon\}$, $C=\{\epsilon\}$, $B=\{\epsilon\}$, and $M(\epsilon,\epsilon)=\specialsymbol$ (where $\specialsymbol$ is the mentioned special symbol). %\footnote{We follow the convention that inputs occur before outputs, thus it is meaningless to provide an output to the empty word.}%\df{I guess we can work with $M(\epsilon,\epsilon)=2^\Gamma$ }

It then processes all tables in the list simultaneously (in BFS), making a small step (e.g. a procedure call) in one, and moving to the next one. 
It proceeds so until an $\scq$ query is answered ``true''.
If the table $\sytb{H}$ is closed, 
\SFour\ extracts a symbolic transducer $\aut{A}$ from it and calls $\scq(\aut{A})$. If the result is $\true$, \SFour\ returns $\aut{A}$.
(Note that $\aut{A}$ is symbolic. If one is interested in a concrete transducer,
any concretization of $\aut{A}$ can be taken instead, since by Claim~\ref{clm:subtransducer} any concretization of $\aut{A}$ is subsumed by $\aut{A}$ and thus is subsumed by $U$ as well.) 
Otherwise, it receives a counterexample $v\in({\ialph\times\oalph})^*$ which is 
a word in $\semafin{\aut{A}}$ but not in $\prefixes{U}$.
It then finds a shortest prefix $u$ of the given counterexample $v$, and adds all suffixes and all prefixes of $\project{u}{\Sigma}$ to $C$. 
	 
%\noindent\makebox[.48\textwidth]{
%	\includegraphics[scale=0.285]{sections/algorithms/sfour.png}
%}

If $\sytb{H}$ is not closed, \SFour\ fills in the missing entries (entries with $\emptysymb$) in $M$
gradually (see Alg.~\ref{alg:filltb}), using $\smq$ calls as needed. 
That is, an entry $M(r,c)$ for $r\in R$ and $c\in C$
is filled only after $M(r',c')$ was filled for every $r',c'\in \ialph^*$ for which $r'c'\prec rc$. 
To fill in the missing entries it performs $\smq$ queries. 
The entries $M(\isym,\epsilon)$ and $M(\epsilon,\isym)$ are filled with $\smq(\biword{\isym}{\qm})$. 
To fill in entries of the form $M(r,c\isym)$ for $|rc\isym|>1$ it performs the query
$\smq(\sytb{H}(rc)\cdot\biword{\isym}{\qm})$. (Recall that $\sytb{H}(rc)$ returns a symbolic word $\iw\oplus\ow$ s.t. 
$\iw=rc$ and $\ow[i]$ is the set of outputs in the entry $M(\iw[i])$.)
If the answer is not $\emptyset$, it fills the answer in. 
Otherwise, the answer $\emptyset$ is accompanied with a counterexample $\w$ which is passed in a call to procedure $\Split$. 
%$$ providing it the input word $rc\isym$.
 Once the table is filled, it calls the procedure $\rebase$.

\input{sections/algorithms/fill.tex}

%\noindent\makebox[.48\textwidth]{
%	\includegraphics[scale=0.285]{sections/algorithms/fill.png}
%}
\input{sections/algorithms/split.tex}

%\noindent\makebox[.48\textwidth]{
%	\includegraphics[scale=0.285]{sections/algorithms/split.png}
%}

The procedure $\Split$, on input $\w=i_1 i_2 \ldots i_{m+1}\oplus o_1 o_2 \ldots o_{m+1}$ %\biword{i_1 }{o_1}\biword{i_2}{ o_2} \ldots \biword{i_m}{ o_m}\biword{i_{m+1}}{o_{m+1}}$ for $i_{m+1}=\isym$
 works as follows (see Alg.~\ref{alg:split}).
%It starts by invoking $\ssq(\sytb{H}(rc)\cdot \biword{\isym}{\true})$.\df{change to use \smq\ that return a cex.}
%Let $\tree{T}_\sytb{H}(rc)=\w$ and assume $\w=\biword{i_1}{\theta_1}\biword{i_2}{\theta_2}\ldots \biword{i_m}{\theta_m}$. 
%Since there was no possible completion to $\w\cdot\biword{\sigma}{\qm}$ in $\prefixes{L}$, the result of the $\ssq$ must be a counterexample $\w'$.
%Assume $\w'=\biword{i_1 }{o_1}\biword{i_2}{ o_2} \ldots \biword{i_m}{ o_m}\biword{i_{m+1}}{o_{m+1}}$ for $i_{m+1}=\isym$ and $o_1,o_2,\ldots,o_{m+1}\in\Gamma$.
It first finds the shortest prefix $\w'$ of $\w$ that is also a counterexample. Let
$\ell$ be the length of $\w'$. The algorithm removes $\sytb{H}$ from $\set{H}$,
creates $\ell$ copies $\sytb{H}_1, \ldots, \sytb{H}_\ell$ of the current table $\sytb{H}$
and makes the following changes in them. 
In a table $\sytb{H}_k$, for $1 \leq k \leq \ell$,
given $M_k(i_1i_2\ldots i_k)=\theta_k$, it checks whether  $\theta_k \setminus  o_k$ is non-empty.
If so,  it updates
all entries corresponding to $\iw = i_1i_2\ldots i_k$ to $\theta_k \setminus  o_k$.
In addition, for every entry which is a suffix of  $\iw$, it deletes the content of the entry, i.e., sets it to $\emptysymb$. (Thus, the entry will be refilled using an \smq\ that takes into account the 
revised value for the prefix $i_1i_2\ldots i_k$). Finally, it adds $\sytb{H}_k$ to $\set{H}$.
If however, $\theta_k \setminus  o_k$ is empty, then $\sytb{H}_k$ is considered \emph{infeasible} and is not added to $\set{H}$.
The case where $k={m+1}$ (which occurs if $\w = \w'$ and hence $\ell = m+1$) is somewhat different,
since the entry $M(i_1i_2\ldots i_k)$ has not been filled yet. Thus $\theta_k$ is taken to be the most general, namely $\Gamma$, the entire set of outputs. 
Note that if an entry's value in $\sytb{H}$ was $\theta$ and its new value in $\sytb{H}_k$ is $\theta'$, then $\theta'\subseteq \theta$.

The procedure $\rebase$ finds a minimal set of rows $B$ that covers all rows $R$ and for each $r\in R$ 
keeps the information of which base rows cover it
(see Alg.~\ref{alg:rebase} in the supplementary material).
%(see Alg.~\ref{alg:rebase} in App.~\ref{app:algs}).
It does so by going over all pairs of rows, checking if one covers the others and recording the information.
The basis is set to the subset of rows which are not covered by any row. The cover set is restricted
 to the rows in the basis.
%Whenever a row is added to the basis, it is searched in the basis for the existence of rows that are implied by $b$,
%and removes any such rows from the basis. This process is called \emph{rebasing}.

\commentout{
Consider two words $r\in R$ and $c\in C$.  Assume $r=\sigma_1\sigma_2\ldots\sigma_k$ and $c=\sigma_{k+1}\sigma_{k+2}\ldots\sigma_\ell$.
The symbolic word $\sytb{H}$ associated with $rc$, denoted $\sytb{H}(rc)$ is defined as follows:
\[\sytb{H}(rc)=\sigma_1\sigma_2\ldots\sigma_\ell \oplus M(\sigma_1)M(\sigma_1\sigma_2)\ldots M(\sigma_1\sigma_2\ldots\sigma_\ell).\]
Recall that $R$ and $C$ are prefix closed, and for any two entries $M(r_1,c_1)$ and $M(r_2,c_2)$ such that $r_1c_1=r_2c_2$ we have that $M(r_1,c_1)=M(r_2,c_2)$. Thus, $\sytb{H}(rc)$  is well defined.
}

%A concrete finite tree $\tree{T}=\tuple{T,\tau}$  is said to be \emph{covered} by $\sytb{H}$ if $T\subseteq RC$ 
%and for every $\iw\in T$ we have that $\tau(\iw)\in \sytb{H}(\iw)$.

%% RBNOTE: This is trivial, can be shortened radically to the following
%% +
\subsubsection{Finding the shortest counterexample}
We can replace a counterexample $\iw\oplus \ow$ by the shortest prefix $\iw' \oplus \ow' \prec \iw\oplus \ow$ for which $\mq(\iw' 
\oplus \ow') = \text{``no''}$.

\commentout{
\subsubsection{Finding the shortest counterexample.}
Suppose the returned counterexample is $\iw\oplus \ow$. 
Assume $\iw=\isym_1\isym_2\ldots \isym_\ell$ and $\ow=\osym_1\osym_2\ldots\osym_\ell$.
%Then there exists a sequence of states  $s_0,s_1,s_2,\ldots,s_{\ell+1}$ in the extracted transducer where $s_0=\epsilon$, $s_{i+1}\in\delta(s_i,\sigma_{i+1})$ and $\osym_i\in\eta(s_{i})$. 
The procedure below is invoked on $\iw\oplus \ow$ and finds the shortest prefix of it that is a counterexample.
Consider the following sequence of  queries 
\[\begin{array}{l}
\mq({\epsilon}),    \\ 
\mq(\biword{\isym_1}{\osym_1}    ),\\ 
\mq(\biword{\isym_1}{\osym_1}\biword{\isym_2}{\osym_2}  ), \\ 
\qquad \ldots\\
\mq( \biword{\isym_1}{\osym_1}\biword{\isym_2}{\osym_2}\ldots\biword{\isym_{\ell}}{\osym_{\ell}}). \\
\end{array}\]
The answer to the first query is ``yes'' since $U$ is non-empty and prefix closed, and the answer to the last query is ``no'' since this is exactly the counterexample. 
Let $k$ be the first index for which the answer is ``no". Then $(\iw\oplus\ow)[1..k]$ is the shortest prefix of $\iw\oplus\ow$, which is
a counterexample.}

\subsubsection{Counterexample processing optimization}
For a shortest counterexample  $\iw\oplus \ow$, the algorithm adds all suffixes and all prefixes of $\iw$ to the columns of the table.

Below we argue that this process will eventually lead to termination.
More precisely, we can point to one of the suffixes that will either reveal a new state in the table,
or remove a transition or make the table infeasible. Thus, it suffices to add this suffix alone to the table.\footnote{This generalizes \cite{RivestS1993}'s optimization of \lstar.}
Recall that ${\iw[1..k]=\isym_1\isym_2\ldots \isym_k}$ and $\ow[1..k]={\osym_1\osym_2\ldots\osym_k}$.
Then there exists a sequence of states  $s_0,s_1,s_2,\ldots,s_{k+1}$ of the extracted transducer for which  $s_0=\epsilon$, $s_{i+1}\in\delta(s_i,\sigma_{i+1})$ and $\osym_i\in\eta(s_{i})$.

Consider then the following sequence of $\smq$  queries and their answers:\footnote{Recall that states are elements of $\Sigma^*$.} 
\[\begin{array}{l@{~\cdot~}ll}
\smq( \epsilon & \biword{\isym_1}{\osym_1} \biword{\isym_2}{\osym_2}\biword{\isym_3}{\osym_3}\ldots \biword{\isym_{k-1}}{\osym_{k-1}} \biword{\isym_k}{\qm})&=\theta_0\\ 
\smq( \sytb{H}(s_1) & \biword{\isym_2}{\osym_2}\biword{\isym_3}{\osym_3}\ldots \biword{\isym_{k-1}}{\osym_{k-1}}\biword{\isym_k}{\qm}))  &=\theta_1\\ 
\smq( \sytb{H}(s_2) & \biword{\isym_3}{\osym_3} \ldots \biword{\isym_{k-1}}{\osym_{k-1}}\biword{\isym_k}{\qm}))  &=\theta_2\\ 
\qquad \ldots\\
\smq( \sytb{H}(s_{k-1}) & \biword{\isym_k}{\qm})) &=\theta_{k-1}\\
\end{array}\]
That is, $\theta_i$ is the result of the query regarding the state we reach upon reading the inputs on the prefix of length $i$ and the respective
answers of our transducer, concatenated to the suffix starting at $i+1$, where the last output is omitted and queried about.
Let $\theta_k=\{\osym_k\}$.

%Suppose the result to the $i$-th such query is $\theta_i$. 
Then $\osym_k\notin\theta_0$, as otherwise $\biword{\isym_1}{\osym_1}\biword{\isym_2}{\osym_2}\ldots\biword{\isym_{k}}{\osym_{k}}$ would not be a counterexample.
%On the other hand, $\osym_k\models\theta_k$.
%To see why this is true, note that (i) $M(s_k)$ implies $M(s_{k-1}\isym_k)$ and  (ii) $\osym_k\models M(s_k,\epsilon)$ since $s_k$ allows emitting $\osym_k$. It follows that  $\osym_k\models M(s_{k-1}\isym_k)=\theta_k$. 
On the other hand, clearly ${\gamma_k\in\theta_k}$.
Hence, for the first index we have ${\osym_k\notin \theta_0}$ and for the last index we have that ${\osym_k\in \theta_k}$. 
Let ${1\leq i \leq k}$ be the first index for which ${\osym_k\notin \theta_i}$ and ${\osym_k\in \theta_{i+1}}$. We then add the column $c=\isym_{i+1}\ldots \isym_k$ to the current table.
Consider the entries $M(s_{i-1}\isym_{i},c)=\theta$ and $M(s_{i},c)=\theta'$.
Then we have $\theta_i \subseteq \theta$ and $\theta_{i+1}\subseteq \theta'$.
From $\osym_k\not\subseteq \theta_i$ it follows that $\osym_k\not\subseteq\theta$.
From $\osym_k\in\theta_{i+1}$ it follows that $\osym_k\in\theta'$. Therefore $\theta' \not\subseteq \theta$.
Before adding column $c$ to the table we had that $M(s_i)$ implies $M(s_{i-1}\isym_i)$ (as otherwise $s_0,s_1,\ldots s_{k}$ won't be a valid run on $\isym_1\ldots\isym_k$). 
Now row $c$ breaks this implication (since $M(s_i,c)\nsubseteq M(s_{i-1}\isym_i,c)$). 

%% RBNOTE: This is far too vague for my taste. ``Revelation'' is not defined, neither is ``infeasibility''
%% + clarifications added
\begin{claim}\label{clm:conterexample-converg}
	We claim that after the counterexample processing (in line~\ref{line:rebase} of Alg.~\ref{alg:sfour}) one of the following happens to the current table:
	\begin{enumerate}
		\item %At least one new state is revealed. That is, at least one row is no longer covered by other rows.
		      At least one row $r\in R \setminus B$ is no longer covered by $B$.
		\item For at least one row $b\in B$ and one letter $\isym\in\ialph$, the set of rows covering $b\isym$ is smaller.
		\item The table becomes infeasible and is removed from the set of tables $\set{H}$.		
	\end{enumerate}
\end{claim}

\section{Correctness and complexity}\label{sec:correctness}% of the algorithm}\label{sec:correctness}

Before we provide the correctness and complexity results we introduce
the measures we use to state them.

\subsubsection{The Rank Measure}
The complexity of the \lstar\ algorithm is defined with respect to the \emph{rank} of the target
language, which is the number of states of the minimal DFA for the language, or equivalently
the number of states in the right congruence relation $\sim_L$.\footnote{For two finite words $u,v$ the relation $u\sim_L v$ holds iff
$uw\in L \Leftrightarrow vw\in L$ for every $w\in\Sigma^*$. }

We define a similar right congruence relation for exhaustive prefix-closed bi-languages.
Let $U$ be an exhaustive bi-language over $\Sigma\times\Gamma$.
For $v_1,v_2\in\Sigma^*$ we say that $v_1 \sim_U v_2$
if for every $w\in(\Sigma\times\Gamma)^*$, $u_1\in U_{|v_1}$,
$u_2\in U_{|v_2}$, we have that $(v_1\oplus u_1)\cdot w\in U$
\emph{iff} $(v_2\oplus u_2)\cdot w\in U$.
We use $\rank{U}$ for the rank of $U$.

There is an additional complexity measure that we need to define on
our target language in order to analyze termination and complexity
of our algorithm. This is the \emph{incompatibility measure} defined as follows.

\subsubsection{The Incompatibility Measure} %Compatible vs. Conflicting Trasducers}
Let $U$ be a target language and $\aut{S}_1,\aut{S}_2$ be symbolic transducers embedded in $U$.
A word $\iw=\isym_1\isym_2\ldots\isym_m\isym_{m+1}$ is said to \emph{witness the incompatibility} of $\aut{S}_1,\aut{S}_2$  wrt. $U$ if 
given ${\given{\sema{\aut{S}_i}}{\iw}=\theta^i_1\theta^i_2\ldots\theta^i_m}$ for $i\in\{1,2\}$
and letting $\theta_j=\theta^1_j \cup \theta^2_j$ for $1\leq j\leq m$,
the result of $\smq(\isym_1\isym_2\ldots\isym_m\isym_{m+1}\oplus \theta_1\theta_2\ldots\theta_m\qm )$ is $\emptyset$, and the result for any respective prefix is not $\emptyset$.
The transducers $\aut{S}_1,\aut{S}_2$ are said to be \emph{incompatible} (or \emph{conflicting}) wrt. $U$  if there exists a word witnessing their incompatibility, otherwise
 they are said to be \emph{compatible}.
 
 \begin{claim}\label{clm:distin-words}
	If $\{\aut{C}_1,\ldots,\aut{C}_m\}$  are pairwise incompatible
	then there exists a word $\iw$ witnessing their incompatibility of size at most $|Q_1|\times|Q_2| \ldots \times |Q_m|$,
	where $Q_i$ is the set of states of $\aut{C}_i$.
\end{claim}

Let $U$ be a target language. 
We use $\incompatible{U}$ for the maximal number of transducers that are embedded in $U$ and are pairwise incompatible wrt. $U$.

	We note that while a deterministic symbolic transducer may output one of a set of outputs in a given state,
	it cannot model two conflicting transducers simultaneously, while a non-deterministic transducer (which may move to different states on the same input) can. 
	The following claim states that the non-deterministic transducers that the algorithm
	returns are consistent and conform only to compatible transducers.

\begin{claim}\label{clm:returns-consistent-trans}
	The procedure $\extractAut(\sytb{H})$  returns a consistent transducer. %=(R,C,M,B)
\end{claim}

\begin{proof}[Proof Sketch]
    Let $\aut{A}_{\aut{H}}=(\ialph,\oalph,Q,q_\iota,\delta,\eta)$ be the extracted transducer.
	%To be consistent it needs to corresponds to a $2^\Gamma$-labeled $\Sigma$-tree.
	We show that there exists a $2^\Gamma$-labeled $\Sigma$-tree $\tree{T}$ 
	that agrees with $\aut{A}_{\aut{H}}$ on every word by induction on the depth of the tree.
	For the root, we label the node $\specialsymbol$.
	Assume the labels of tree agree with the output of the transducer for every word $v\in\ialph^*$
	of length $\ell$. Consider such a word $v$ and its one letter extension $v\isym$. 
	Assume $\delta(q_\iota,v\sigma)=Q$.
	Recall that states in $\aut{A}_{\aut{H}}$ are words in $\Sigma^*$ that correspond
	to rows in $\aut{H}$. 
	We set $\tree{T}$ to label the node  $v\isym$ by the union $\cup_{r\in Q}M(r,\epsilon)$. 
	Again, by induction on the length of a word, the tree agrees with the transducer on every input word.
\end{proof}

\subsubsection{Termination and Complexity Results}
We prove the termination and complexity results gradually.

First, we observe that if $U$ exactly corresponds to a language of a concrete transducer, then the
algorithm performs exactly as \lstar\ for Moore machines.
\begin{lemma}\label{lem:com-one-lang}
	If the unknown bi-language $U$ contains words corresponding to a single concrete transducer then the algorithm never maintains more than one table, and it terminates in time polynomial in the number of states of the transducer.
\end{lemma}

Next, we discuss the case where $U$ does not contain conflicting implementations,
namely $\incompatible{U}=1$, but there might be several (compatible) transducers $\aut{S}$
for which $\sema{\aut{S}}\subseteq U$.

\begin{lemma}\label{lem:no-conflicts}
	Assume $\incompatible{U}=1$, $\rank{U}=n$ and $\ell$ is the size of the maximal counterexample received by the algorithm.
	Then the algorithm never maintains more than one table; it terminates in time polynomial in $n$ and $\ell$; the number of $\scq$s asked is bounded by $O(n|\Sigma|)^2$;
	and the number of $\smq$s is bounded by $O(\ell(n|\Sigma|)^3)$.
\end{lemma}
\begin{proof}[Proof sketch]
	In this case the symbolic membership queries always produce non-empty set of output expression, therefore no splitting of tables will occur. 
	The size of the basis is bounded by the rank $n$. Therefore the number of rows is bounded by $n+n|\Sigma|$.
	The number of columns is determined by the number of counterexamples received for an $\scq$. %This number is bounded by the number of $\scq$s.
	Since a counterexample leads to adding a new state to the basis or eliminating at least one implication (see Claim~\ref{clm:conterexample-converg}), and the number of implications 
	is bounded by a square of the number of rows, the number of \scq s is at most $(n+n|\Sigma|)^2$.
	Given that the maximal length of a counterexample is $\ell$, and for each counterexample we add all prefixes and suffixes,
	 the number of columns is bounded by $2\ell(n+n|\Sigma|)^2$.
\end{proof}

We are now ready to discuss the most general case, where $\incompatible{U}=m$ and $\rank{U}=n$.

\begin{theorem}\label{thm:complexity}
If $\incompatible{U}=m$ and $\rank{U}=n$, and $\ell$ is the size of the maximal counterexample received by the algorithm, then the algorithm
(Alg.~\ref{alg:sfour}) terminates in time polynomial in $m$, $n$ and $\ell$. The number of $\scq$s asked is bounded by $O(mn|\Sigma|)^2$
	and the number of $\smq$s is bounded by $O(m\ell(n|\Sigma|)^3)$.
\end{theorem}
\begin{proof} 
In this case since there are $m$ incompatibilities, there will be at most $m$ splits of tables. Note that the algorithm essentially performs a BFS
on these tables. The number of leaves in the spanned tree is bounded by $m$, therefore its size is $O(m)$.
For each leaf, the processing time is bounded by a polynomial in $n$ and $\ell$ as per Lemma~\ref{lem:no-conflicts}.
Therefore the overall number of steps is bounded by a polynomial in $m$, $n$ and $\ell$,
and the number of queries is at most $m$ times the number of queries as per Lemma~\ref{lem:no-conflicts}.
\end{proof}

As in \lstar, the algorithm may not converge if the target is non-regular, namely,
if its rank is infinite. We show that given $\rank{U}$ is finite, our algorithm will converge.
First we state that if $\rank{U}$ is finite, so is $\incompatible{U}$.

\begin{lemma}\label{lem:finite-incompatability}
    If $\rank{U}$ is finite, then $\incompatible{U}$ is finite as well.
\end{lemma}

It follows from Theorem~\ref{thm:complexity} and Lemma~\ref{lem:finite-incompatability} that the fact that $\rank{U}$
is finite suffices to guarantee termination.

\begin{corollary}\label{cor:termination}
	If  $\rank{U}$ is finite 
	then
	the algorithm terminates and returns a contained transducer.
\end{corollary}	%\quad\\

\par\label{par:irregular}
While Corollary~\ref{cor:termination} proves that Alg.~\ref{alg:sfour} terminates conditioned $U$ has a finite rank, namely it
is a regular language, our algorithm may terminate also for some target languages $U$ that 
contain a non-regular language. In particular, it terminates and returns a valid transducer
for the target language $U = L' \cup L'' \cup L'''$ over
$\Sigma=\{a,b\}$, $\Gamma=\{0,1,2\}$, 
defined as follows:
\[
\begin{array}{l@{\,=\,}ll}
L_n & \{a^n b^n w ~|~  w\in\Sigma^\omega \} & \text{for } n\in\mathbb{N}\\
L'  & \{ w_n \oplus 0^{2n-1}10^\omega ~|~ n > 0, w_n \in L_n \}, \\
L'' & \{ w\oplus 0^\omega ~|~ \forall n > 0, w \notin L_n \}, \\
L''' & \Sigma^\omega \oplus 2^\omega.
\end{array}
\]
The learned symbolic transducer
generates the $\omega$-regular language $L'''$.
This is because Alg.~\ref{alg:sfour} traverses the tree of possible
implementation using BFS, and terminates once one of the branches converged. 
%\end{remark}

%\subsection{Complexity}

%\newpage
\section{Experimental Results}\label{sec:experminetal}
 Table~\ref{tab:all-results}
presents the results on non-trivial arbiters (systems granting requests), showing $|\Sigma|$, $|\Gamma|$, the number of performed queries (\mq, \smq, \scq), the total number of generated tables, the number of tables analyzed in \SFour{}'s main loop ($|\mathbb{H}_{\textrm{gen}}|$, $|\mathbb{H}_{\textrm{exp}}|$), the number of  splits (\#split), the number of states of the extracted transducer $|\mathcal{A}|$ ($-$ for unrealizable), the learning time of \SFour\ and the oracle time in seconds.\footnote{We implemented Alg.~\ref{alg:sfour} in C++$17$ and used Spot 2.8.4 (\url{https://spot.lrde.epita.fr/}) for representing LTL formulas and $\omega$-automata. More details on the experiments are provided in appendix~\ref{app:imp-details}.
The implementation is available on GitHub and will be made public after notification.  The experiments were executed on an Intel\textregistered~Core\texttrademark{} i7-7567U CPU @ $3.50$GHz CPU with $16$GB RAM compiled with Clang $11$ on MacOS Catalina $10.15.6$.} More details on the implementation and experiments can be found in App.~\ref{app:imp-details}.

%Note that the learning time of \SFour\ indicates the
%performance of our algorithm (as oracles are external).

\begin{table}[!ht]
	\centering
	\scriptsize
	\setlength{\tabcolsep}{2pt} % reduce column width
	\renewcommand{\arraystretch}{0.9} % reduce row height
	\resizebox{\columnwidth}{!}
	{
		\begin{tabular}{ccc cccc cccc cc}
			\toprule
			\multicolumn{3}{c}{\textbf{Experiment}} &
			\multicolumn{3}{c}{\textbf{Query stat.}} &
			\multicolumn{4}{c}{\textbf{Synth. stat.}} &
			\multicolumn{2}{c}{\textbf{Time}} \\
			\cmidrule(lr){1-3}
			\cmidrule(lr){4-6}
			\cmidrule(lr){7-10}
			\cmidrule(lr){11-12}
			No. & $|\Sigma|$ & $|\Gamma|$ &
			MQ & SMQ & SCQ &
			$|\mathbb{H}_{\textrm{exp}}|$ & $|\mathbb{H}_{\textrm{gen}}|$ & \#\textit{split} & $|\mathcal{A}|$ &
			\SFour & Oracle \\
			\midrule
			\begin{filecontents*}{all-results.csv}
experiment,   Sigma,   Gamma,    MQ,  SMQ, SCQ,    Hexp,    Hgen,  splits,   A,   SFour,  Oracle
         1,       2,       2,     14,    6,   1,       1,       1,       0,   2,    0.00,    0.00
         2,       2,       2,     18,   11,   1,       2,       3,       1,   3,    0.00,    0.01
         3,       2,       2,   2563,   52,   2,       1,       1,       0,   2,    0.00,    4.36
         4,       1,       2,   3316,   10,   2,       1,       1,       0,   4,    0.00,    0.06
         5,       2,       2,      5,    4,   1,       1,       1,       0,   2,    0.00,    0.00
         6,       2,       4,     20,    6,   1,       1,       1,       0,   2,    0.00,    0.00
       7-1,       2,       2,      5,    4,   1,       1,       1,       0,   2,    0.00,    0.00
       7-2,       4,       4,     62,   41,   1,       2,       9,       1,   3,    0.00,    0.01
       7-3,       8,       8,   6233, 4596,   1,      27,     399,       4,   4,    0.85,    0.69
       8-1,       2,       2,     75,   26,   2,       1,       1,       0,   4,    0.01,    0.01
       8-2,       4,       4,   1135,  138,   2,       1,       1,       0,   8,    0.05,    0.72
       8-3,       8,       8,  19235, 1010,   2,       1,       1,       0,  14,    3.60,   31.28
       8-4,      16,      16, 283809, 8034,   2,       1,       1,       0,  22, 1557.65, 1001.92
       9-1,       2,       2,     18,   12,   0,       3,       3,       2,   0,    0.00,    0.00
       9-2,       2,       2,    848,  180,   4,      11,      11,       5,  10,    0.16,    0.20
       9-3,       2,       2,   3982,  397,   4,      18,      18,       6,  20,    0.34,    0.95
       9-4,       2,       2,  21433, 1227,   5,      73,      84,      17,  38,    1.66,    7.45
       9-5,       2,       2,  87866, 2774,   6,      95,      95,      14,  72,    9.08,   62.00
       9-6,       2,       2, 534221,  7225,   6,     221,     291,      28, 138,   40.96, 1189.78
       9-7,       2,       2, 1996801, 17083,  7,     483,     483,      48, 268,    0.00, 16519.96
      10-1,       2,       2,    277,   170,   6,      52,      52,       9,   6,    0.10,    0.02
      10-2,       2,       2,    211,   121,   4,      19,      28,       4,   7,    0.10,    0.02
      10-3,       2,       2,   2091,  1070,   9,    1115,    1115,      40,   9,    7.43,    0.11
      10-4,       2,       2,  36042, 19391,   7,   40850,   40850,     155,  12,  502.45,    2.14
\end{filecontents*}
\csvreader[head to column names, late after line=\\]{all-results.csv}{}%
			{\experiment & \Sigma & \Gamma & \MQ & \SMQ & \SCQ & \Hexp & \Hgen & \splits & \A & \SFour & \Oracle}%
			\bottomrule
		\end{tabular}
	}
	\caption{All experimental results. The target languages of experiments 7 to 10 are parameterized. 
		For these experiments the number after dash represents the argument value of the corresponding parameter.
		%\textcolor{red}{This table is subject to change.}
	}
	\label{tab:all-results}
\end{table}

%The results in Table~\ref{tab:all-results} are produced by linking our
%implementation with Spot 2.9.4 library and executing it in a Docker container with Ubuntu Linux 20.04, 
%to closely simulate the conditions in the submitted code. This setup differs slightly from the
%execution environment used in the main body of the paper.\footnote{We observed some variation in
%	performance for a few examples; this is due to external factors and is to be expected.}

The set of examples 9-1, $\ldots$, 9-7
demonstrates that our algorithm can generate transducers with hundreds of states. 
The set of examples 10-1, $\ldots$, 10-4 demonstrates that the number of splits can be high even if
the resulting transducer is relatively small. 

\section{Discussion}\label{sec:discuss}
We introduced a new problem, of constructing an implementation of a reactive system
without being given a formal specification as in reactive synthesis. Instead we assume that
we have knowledge about good and bad behaviors, specifically we use 
 symbolic membership queries and symbolic conjectures queries. The problem is motivated
 by real scenarios of inferring environment for systems that need to work in 
 heterogeneous third-party environments.\footnote{In a practical setting, membership queries may be implemented via executing a set of black boxes, conjecture queries may be replaced by massive random membership queries, and the result we obtain does not fall under \emph{exact learning}, but rather under \emph{PAC learning}~\cite{Valiant13,Angluin87}.}

 We have shown that given the target language has
 a finite rank, our algorithm terminates, and its time and query
 complexity is polynomial with respect to the target language's rank, its incompatibility measures 
 and the size of the longest counterexample. We note that in cases where the target language's
 incompatibility measure is high, the outputted transducer may still be small (that is, its number of states may be much smaller than the target language's rank).
 For future research we would like to investigate the problem of finding the incompatibility measure
 of a language defined by a given temporal logic specification or a given $\omega$-automaton.

%% file: sections/algorithms/sfour.tex
\begin{algorithm}[tb]
	\caption{$\SFour$.}\label{alg:sfour}
\begin{algorithmic}[1]
    %\Require{$\smq_U$,$\ssq_U$,$\scq_U$,$\mq_U$ for an exhaustive bi-language $U$}
    %\Ensure{Automaton $\aut{A}$ implementing $U$}
    \Statex
    \Function{$\SFour$}{$\smq_U$, $\scq_U$}
        \Let{$\sytb{H}$}{$(R\!\leftarrow\!\{\epsilon\},C\!\leftarrow\!\{\epsilon\},B\!\leftarrow\!\{\epsilon\}, M(\epsilon,\epsilon)\!\leftarrow\!\specialsymbol)$}
        \Let{$\set{H}$}{$\{ \sytb{H} \}$} 
        \While{$\set{H} \neq \{\}$}
    	    \ForAll{$\sytb{H} \in \set{H}$}
                \If{$\isClosed(\sytb{H} )$}
        		     \Let{$\aut{A}$}{$\extractAut(\sytb{H} )$}
        		     \Let{$\iw$}{$\scq(\aut{A})$}
        		     \If {$v = \epsilon$}
        		     \State \Return{\aut{A}}
        		     \Else
        		     \Let {$u$}{$\findShortestCE(v)$}
        		     \Let {$C$}{$C\cup \suffixes{u}\cup\prefixes{u}$}
        		     \EndIf
                \Else 
                    \State {$\fillTb(\set{H},\sytb{H} )$}\label{line:fill}
                    \State {$\rebase(\sytb{H})$}\label{line:rebase}
                \EndIf
            \EndFor  
        \EndWhile
        \State \Return{unrealizable}
    \EndFunction
\end{algorithmic}
\end{algorithm}

%% file: sections/algorithms/fill.tex
\begin{algorithm}
	\caption{$\fillTb$.}\label{alg:filltb}
\begin{algorithmic}[1]
    %\Require{Bi-language $L$, List of Symbolic Tables $\set{T}$, Symbolic Table $\mathrm{T}=(R, C, M, B)$}
    %\Ensure{Automaton $\aut{A}$ implementing $L$}
    \Statex
    \Procedure{$\fillTb$}{$\set{H},\sytb{H}=(R, C, M, B)$}
    \Let{$E_1$}{$\{rc~|~r\in R, c\in C, M(r,c)=\emptysymb\}$}
    \Let {$E_2$}{$\{b\isym c~|~c\in C,\isym\in\ialph,b\in B, b\isym\notin R\}$}
    \Let {$R$}{$R\cup \{b\isym~|~ b\in B, b\isym\notin R\}$}
    \Let {$E$}{$E_1\cup E_2$} \label{line:E:alg:fill}
    \While{$E\neq\emptyset$}
	    \Let {$e\isym$}{a shortest prefix in $E$}
	    \Let {$E$}{$E\setminus \{e\isym \}$}
	    \Let {$(\theta,\w)$}{$\smq(\sytb{H}(e)\biword{\isym}{\qm})$}
	    \If {$\theta\neq\emptyset$}
		    \ForAll {$r\in R$, $c\in C$ s.t. $rc=e\isym$}
		    \Let{$M(r,c)$}{$\theta$} 
		    \EndFor
	   \Else
			 \State {$\Split(\set{H},\sytb{H},\w)$}   
			 \State \Break
	   \EndIf	 
    \EndWhile
    \State \Return
    \EndProcedure
\end{algorithmic}
\end{algorithm}

%% file: sections/algorithms/split.tex
\begin{algorithm}[tbh]
	\caption{$\Split$.}\label{alg:split}
\begin{algorithmic}[1]
    %\Require{$\smq_U$,$\ssq_U$,$\scq_U$,$\mq_U$ for an exhaustive bi-language $U$}
    %\Ensure{Automaton $\aut{A}$ implementing $U$}
    \Statex
    \Procedure{$\Split$}{${\set{H}},{i_1 i_2 \ldots i_{m+1}\oplus o_1 o_2 \ldots o_{m+1}}$} 
    \Let{$\ell$}{$|\findShortestCE(\w)|$}   
    %\Let{$\w''$}{$\biword{i_1}{o_1}\biword{i_2}{o_2}\ldots \biword{i_\ell}{o_\ell}$} 
    %\Let{$\w''$}{$i_1 i_2 \ldots i_{\ell}\oplus o_1 o_2 \ldots o_{\ell}$}
    \Let{$\set{H}$}{$\set{H} \setminus \{ \sytb{H}\}$} \label{line:H:alg:split}
    \For{$1\leq k \leq \ell$}
        \Let{$\sytb{H}_k$}{$\sytb{H}$}
        \State{\text{Assume }{$\sytb{H}_k=(C_k,R_k,M_k,B_k)$}}
        \If{$k\neq m+1$}
            \Let{$\theta_k$}{$M_k(r,c)\setminus o_k$}
        \Else
            \Let{$\theta_k$}{$\Gamma\setminus o_k$}
        \EndIf
        \If{$\theta_k \neq \emptyset$}
        \ForAll{$r\in R,\ c\in C$ s.t. $rc=i_1\ldots i_k$}
            \Let{$M_k(r,c)$}{$\theta_k \setminus o_k$ }%{$\theta_k\setminus o_k$}
        \ForAll{$r'\in R,\ c'\in C$ s.t. $rc \prec r'c'$}
            \Let{$M_k(r,c)$}{$\emptysymb$}
        \EndFor  
        \EndFor  
        \Let{$\set{H}$}{$\set{H} \cup \{ \sytb{H}_k \}$} 
        \EndIf \label{line:endsplit:alg:split}
    \EndFor
        
    \State \Return
    \EndProcedure
 \end{algorithmic}
\end{algorithm}

%% file: sections/algorithms/rebase.tex
\begin{algorithm}
	\caption{$\rebase$.}\label{alg:rebase}
\begin{algorithmic}[1]
\Function{\rm rebase}{\tab{T}} 
	\Let{$B'$}{$\emptyset$}
    \ForAll{$r \in R$}% \setminus B$}
       \Let{\newbase{r}}{$\{\}$}
       \ForAll{$b \in R$}
        %\ForAll{$b \in B'$}
            \If{$\text{covers}(M(b),M(r))$}
                \Let{\newbase{r}}{$\newbase{r}\cup\{ b \}$}
            \EndIf
        \EndFor
        \If{$\newbase{r} = \{r\}$} 
            \Let{$B'$}{${B'}\cup\{r\}$}
        \EndIf
    \EndFor    
    \ForAll{$r \in R$}% \setminus B$}
    \Let{\newbase{r}}{$\newbase{r}\cap B'$}
    \EndFor    
    \State\Return $(R,C,M,B',\nabla')$
\EndFunction
\end{algorithmic}
\end{algorithm}

%% file: sections/omitted_proofs.tex
\subsection{Omitted proofs of Section~\ref{sec:defs}}

\noindent
\claimref{Claim~\ref{clm:exhaustive-not-realizable}}
states the following
\begin{itemize}
    \item []
	\emph{It may be that $L$ is \ialphit-exhaustive yet $\concTrees(L)= \emptyset$.}
\end{itemize}
\begin{proof}
	Take $\ialph=\{0,1\}$ and $\oalph=\{a,b\}$ and consider $L= \{0^\omega\oplus a^\omega \} \cup  \{\iw \oplus b^\omega ~|~  \iw\in\ialph^\omega\setminus 0^\omega\}$. 
	It is easy to see that for every $\iw\in\ialph^\omega$ there exists $\ow\in\oalph^\omega$ such that $\iw\oplus\ow\in L$ thus $L$ is \ialphit-exhaustive. To see why $\concTrees(L)=\emptyset$ note that the only possible label for the path $0^\omega$ is $a^\omega$, whereas the only possible label for the path $01^\omega$ is $b^\omega$, thus no matter how the node $0$ is labeled, we won't be able to satisfy the requirement for tree containment in $L$. 
\end{proof}

\noindent
\claimref{Claim~\ref{clm:conc-imp-symb}} states the following
\begin{itemize}
    \item [] 
    	\emph{If $\concTrees(L)\neq \emptyset$ then $\symbTrees(L)\neq \emptyset$.
    }
\end{itemize}

\begin{proof}%[Proof of Claim~\ref{clm:conc-imp-symb}]
	A concrete-tree is a special type of a symbolic-tree.
\end{proof}

\noindent
\claimref{Claim~\ref{clm:conc-in-symb}} states the following
\begin{itemize}
    \item [] 
    \emph{Let $\tree{T}_S$ be a symbolic-tree in $\symbTrees(L)$. Let $\tree{T}_C$ be a concrete-tree such that $\tree{T}_C(\iw)\in \tree{T}_S(\iw)$ for every $\iw\in\ialph^*$.
	Then $\tree{T}_C\in\concTrees(L)$.
	}
\end{itemize}

\begin{proof}%[Proof of Claim~\ref{clm:conc-in-symb}]
    Suppose not. Then there exists a word $\iw\in\ialph^\omega$ such that for the word $$\ow=
	\tree{T}_C(\iw[0])\cdot \tree{T}_C(\iw[1])\cdot \tree{T}_C(\iw[2])\cdots$$ we have that $\iw\oplus\ow\notin L$. Let $$\alpha=\tree{T}_S(\iw[0])\cdot \tree{T}_S(\iw[1])\cdot \tree{T}_S(\iw[2])\cdots$$ Then by the claim's premise $\ow[i]\in\alpha[i]$ for every $i\in\mathbb{N}$. Contradicting that $L_{|v}\supseteq \{w\in\oalph^\omega~|~\forall i \in \mathbb{N}.\ w[i]\in \alpha[i]\}$.
\end{proof}

\noindent
\claimref{Claim~\ref{clm:trans-fin-inf-safety}}
states the following
\begin{itemize}
    \item [] 
	\emph{$\prefixes{\semainf{\aut{A}}}=\semafin{\aut{A}}$ and $\semainf{\aut{A}}$ is a safety language.
	}
\end{itemize}
\begin{proof}%[Proof sketch]
	The first statement holds since if an $\omega$-word is generated by $\aut{A}$ then so are all its prefixes.
	It follows that  
	$\semainf{\aut{A}}$ is the set of all $\omega$-words all of whose prefixes are in $\semafin{\aut{A}}$.
	By~\cite{MannaP89}, a language $L\subseteq \Sigma^\omega$ is safety iff there exists $S\subseteq \Sigma^*$ such that $L = \{w\in\Sigma^\omega~|~\forall i.~ w[..i]\in S\}$. 
	Take ${S=\semafin{\aut{A}}}$. Thus, ${\semainf{\aut{A}}}$ is a safety language.
\end{proof}

\noindent
\claimref{Claim~\ref{clm:subtransducer}} states the following
\begin{itemize}
    \item [] 
    \emph{Let $\aut{S}=(\ialph,\oalph,Q,q_\iota,\delta,\eta)$ and $\aut{S'}=(\ialph,\oalph,Q,q_\iota,\delta',\eta')$ be symbolic transducers s.t.
    $\delta'(q,\isym)\subseteq \delta(q,\isym)$ and $\eta'(q)\subseteq \eta(q)$ for every $\isym\in\ialph$ and $q\in Q$. Then $\sema{\aut{S}}\subseteq\sema{\aut{S}'}$.
    }
\end{itemize}

\begin{proof}
    Let $\iw\oplus\ow\in\sema{\aut{S}}$ where $\iw=\isym_1\isym_2\ldots\isym_m$ and $\ow=\osym_1\osym_2\ldots,\osym_m$. 
    Then there exists a sequence of states $q_\iota,q_1,q_2,\ldots,q_m$ that is a run of $\aut{S}$ on $\iw$ such that $\osym_i\in\eta(q_i)$ for $1\leq i\leq m$.
    Since $\delta'(q,\isym)\subseteq \delta(q,\isym)$ for every $\isym\in\ialph$ and $q\in Q$ it follows that
    $q_\iota,q_1,q_2,\ldots,q_m$ is a run of $\aut{S}'$ on $\iw$. Since  $\eta'(q)\subseteq \eta(q)$ for every $q\in Q$
    it follows that $\osym_i\in\eta'(q_i)$ for every $1\leq i\leq k$. Therefore $\iw\oplus\ow\in\sema{\aut{S}'}$ as well.
\end{proof}

\subsection{Omitted proofs of Section~\ref{sec:learn}}
\commentout{
The following claim states the followingfor every closed table there exists a unique minimum symbolic transducer that agrees with the table on every entry in the table
\begin{lemma}
	For every closed table 
	$\sytb{H}$ 
	there exists a unique minimum symbolic transducer 
	$\sytb{A}$ 
	such that $\sema{\sytb{H}} \subseteq \sema{\sytb{A}}_*$. %%RBNOTE please check! %% dana: No, its containment.
\end{lemma}

\begin{proof}[Proof Sketch]
	Since implication is a partial order, the set of minimal elements is well defined.
	This set comprises the states of the minimal transducer. The transition relation is defined by
	following the one-letter extension of the rows corresponding to states. Specifically, if row $r$ is a minimal element, and
	the row corresponding to $r\sigma$ is implied by the set of rows $\{r_1,\ldots,r_k\}$, then there are transitions from
	$r$ on $\sigma$ to any state $r_i$ for $1\leq i\leq k$. This shows that the transducer is unique and well defined.
\end{proof}}

\begin{claim}\label{clm:minimal-basis}
    Procedure $\rebase$ (Alg.~\ref{alg:rebase}) returns a minimum basis.
\end{claim}

\begin{proof}
    The procedure computes the implication relation for the set of rows. Specifically, for
    every row $r$ it computes the set of rows $S_r$ that imply it. If $S_r$ is a singleton,
    it must be the singleton $\{r\}$ (because of reflexivity of implication). This means 
    that $r$ is a minimal element in the partial order of implications. Therefore $r$
    must be in the basis. All other rows are not in the basis.
\end{proof}

\begin{claim}\label{clm:rank-and-minimal-trans}
    Let $\aut{S}$ be a consistent transducer, and let $U=\sema{\aut{S}}$. 
    If $\rank{U}=n$ then $\aut{S}$ has at least $n$ states.
\end{claim}

\begin{proof}
    Assume this is not the case. Then there exists two words $v_1,v_2\in\Sigma^*$ s.t. $v_1\not\sim_U v_2$,
    yet $v_1$ and $v_2$ reach the same state $q$ of $\aut{S}$.
    From $v_1\not\sim_U v_2$ it follows that $\exists w\in(\Sigma\times\Gamma)^*$,  
    $u_1\in U_{|v_1}$ and $u_2\in U_{|v_2}$ s.t. wlog. $(v_1\oplus u_1)\cdot w\in U$ and $(v_2\oplus u_2)\cdot w\notin U$.
     Therefore if $\project{w}{\oalph}$ can be emitted from $q$ on reading  $\project{w}{\ialph}$,
     then $\aut{S}$ wrongly accepts   $(v_2\oplus u_2)\cdot w$,
     and otherwise $\aut{S}$ wrongly rejects   $(v_1\oplus u_1)\cdot w$.
\end{proof}

\noindent
\claimref{Claim~\ref{clm:extracted-trans}} states the following
\begin{itemize}
    \item [] \emph{
    Let $\sytb{H}$ be a closed and minimal symbolic table, and $\aut{A}_{\sytb{H}}$ the transducer extracted from it.
    Then  $\aut{A}_{\sytb{H}}$ is one of the minimal transducers that agrees with $\sytb{H}$ and for any other minimal transducer $\aut{A}$ that agrees with $\sytb{H}$
    it holds that $\sema{\aut{A}}\subseteq \sema{\aut{A}_{\sytb{H}}}$. 
    }
\end{itemize}

\begin{proof}
    Assume towards contradiction that $\aut{A}_{\sytb{H}}$  does not agree with the table on some entries.
	Let row $r$ and column $c$ be such that $rc$ is a shortest prefix on which they disagree.
	That is, $M(r,c)=\theta$, yet the transducer $\aut{A}_{\sytb{H}}$, on reading $rc$ reaches a set of states $S=\{s_1,\ldots,s_m\}$ with respective outputs 
	$\theta_i$ and $\bigcup_{\{1\leq i \leq m\}} \theta_i \nsubseteq \theta$.
	Therefore, there exist a state $s_j$ for which $\theta_j\nsubseteq \theta$. 
	Let $rc=\isym_1\ldots\isym_\ell$. 
	Let $s$ be the row corresponding to $\isym_1\ldots\isym_{\ell-1}$.
	Then the row $s_j$ implies the row $s \isym_\ell$. 
	It follows from the fact that $rc$ is a shortest prefix where they disagree, that the output on $s$ agrees with the table. When we fill in the entry for $s\isym_\ell$ the $\smq$ took in account the output of $s$
	and all states reaching it. Specifically the query was $\smq(\sytb{H}(s)\cdot\biword{\isym_\ell}{\qm})$ and the output was $\theta$.
	Recall that $\sytb{H}(s)$ consists of all words $s\oplus\ow$ that agree with the table.
	If $\theta_j\nsubseteq \theta$ then the row $s_j$ does not imply the row $s$, contradicting our assumption on $s_j$.
	
	This shows $\aut{S}$ agrees with the table. The fact that $\aut{S}$ is minimal follows from the fact that 
	the basis is minimal (as per Claim~\ref{clm:minimal-basis}) and $\aut{S}$ consists one state per row in the basis,
	and by Claim~\ref{clm:rank-and-minimal-trans} no transducer with less states accepts this languge.
	
	Clearly among all transducers with the same structure as $\aut{A}_{\sytb{H}}$ the transducer $\aut{A}_{\sytb{H}}$
	has the maximal number of transitions that conform to $\sytb{H}$ and the maximal number of outputs on states that
	conform to $\sytb{H}$. Thus $\sema{\aut{A}}\subseteq \sema{\aut{A}_{\sytb{H}}}$.
\end{proof}

\noindent
\claimref{Claim~\ref{clm:conterexample-converg}}
\begin{itemize}
    \item [] 
	\emph{We claim that after the counterexample processing (in line~\ref{line:rebase} of Alg.~\ref{alg:sfour}) one of the following happens to the current table:
	\begin{enumerate}
		\item %At least one new state is revealed. That is, at least one row is no longer covered by other rows.
		      At least one row $r\in R \setminus B$ is no longer covered by $B$.
		\item For at least one row $b\in B$ and one letter $\isym\in\ialph$, the set of rows covering $b\isym$ is smaller.
		\item The table becomes infeasible and is removed from the set of tables $\set{H}$.		
	\end{enumerate}
	}
\end{itemize}

\begin{proof}
	Recall that the counterexample $\w$ processing procedure first finds a shortest prefix $u$ of the given counterexample that is still
	a counterexample, then adds all its suffixes and prefixes to the columns, and fills in the missing entries using $\smq$s. As always,
	it could be the case that the $\smq$ returns $\emptyset$, in which case the third item holds.
	
	Otherwise, by the discussion at the end of subsection\ref{subsec:alg} (paragraph titled \emph{Counterexample processing optimization}),
	the number of implications is reduced. Hence, either there is at least one less transitions,
	or a row that was implied by another row is no longer implied by any row, and therefore is added to the
	basis. 
\end{proof}

\subsection{Omitted proofs of Section~\ref{sec:correctness}}

\begin{claim}\label{clm:conc-realizes}
	If $\scq(\aut{A})$ returns $\true$, then any concretization $\aut{C}$ of $\aut{A}$ realizes $U$.%\df{Is this how we defined realizes?}
\end{claim}

%\noindent
%\claimref{Claim~\ref{clm:conc-realizes}} 
%states the following
%\begin{itemize}
%    \item [] \emph{
%	If $\scq(\aut{A})$ returns $\true$ then any concretization $\aut{C}$ of $\aut{A}$ realizes $U$.%\df{Is this how we defined realizes?}
%\end{itemize}
\begin{proof}[Proof sketch]
	If $\aut{C}$ is a concretization of $\aut{A}$ then $\sema{\aut{C}}\subseteq\sema{\aut{A}}$. % by Claim~\ref{clm:embeded-tree-subsumed-by-symb-tree}.
	And by the result of the $\scq$ we know that $\sema{\aut{A}}_*\subseteq\prefixes{L}$.
	From Claim~\ref{clm:trans-fin-inf-safety} we know that 	$\prefixes{\semainf{\aut{A}}}=\semafin{\aut{A}}$.
	Thus $\prefixes{\semainf{\aut{A}}}=\semafin{\aut{A}}\subseteq\prefixes{U}$.
	We know that $U$ is prefix-closed since $U$ is safety, and that $\semainf{\aut{A}}$ is prefix closed from  Claim~\ref{clm:trans-fin-inf-safety}.
	It follows that $\semainf{\aut{A}}\subseteq L$ and  $\aut{C}$ realizes $U$.
\end{proof}

\noindent
\claimref{Claim~\ref{clm:distin-words}} 
states the following
\begin{itemize}
    \item [] 
	\emph{If $\{\aut{C}_1,\ldots,\aut{C}_m\}$  are pairwise incompatible
	then there exists a word $\iw$ witnessing their incompatibility of size at most $|Q_1|\times|Q_2| \ldots \times |Q_m|$.
	}
\end{itemize}
\begin{proof}[Proof sketch]
    Consider the product construction of all transducers. If we can label each state in a manner consistent 
    with each of the given transducers then they are compatible. Otherwise, there exists a reachable state in the product
    construction which cannot be labeled in consistency with all. The access word to this state is witnessing their
    incompatibility and its size is at most $|Q_1|\times|Q_2| \ldots \times |Q_m|$.
\end{proof}

\noindent
\claimref{Claim~\ref{lem:com-one-lang}}
states the following
\begin{itemize}
    \item [] 
	\emph{If the unknown bi-language $U$ contains words corresponding to a single concrete transducer then the algorithm never maintains more than one table, and it terminates in time polynomial in the number of states of the transducer.
	}
\end{itemize}
\begin{proof}[Proof sketch]
	In this case every symbolic membership query will be answered by a set of output expressions which is a singleton, and the algorithm will work exactly as the algorithm for learning Moore machines using \mq\ and \eq, which is a trivial extension of \lstar~\cite{Fisman18}.
\end{proof}

\noindent
\claimref{Claim~\ref{lem:finite-incompatability}}
states the following
\begin{itemize}
    \item [] 
    \emph{If $\rank{U}$ is finite, then $\incompatible{U}$ is finite as well.
    }
\end{itemize}

\begin{proof}[Proof sketch]
    If $\rank{U}$ is finite then there exists a non-deterministic symbolic transducer $\aut{S}$
    such that $\sema{\aut{S}}=U$. Let $\aut{S}_1$ and $\aut{S}_2$ be two incompatible transducers wrt. $U$
    and let $\iw=\isym_1\isym_2\ldots\isym_m\isym_{m+1}$ be the word {witnessing their incompatibility}.
    Let $S$ be the set of states in $\aut{S}$ that is reached upon reading $\iw$. For
    $\iw$ to be a distinguishing word it must be that there are two distinct states $s_1,s_2\in S$
    such that $\sema{\aut{S}_1}(\iw)\in\eta(s_1)\setminus\eta(s_2)$ and  $\sema{\aut{S}_2}(\iw)\in\eta(s_2)\setminus\eta(s_1)$.
   Since the number of pairs of states accessible in $\aut{S}$ by the same word is bounded, as is the number of sunsets of $\Gamma$
   (the possible output for $\eta(\cdot)$), so is $\incompatible{U}$.    
\end{proof}

%\subsection*{Additional claims}

The following claim asserts that the algorithm does not lose information when preforming a table split.

%% RBNOTE: Cannot find a definition of table COVERS tree
% dana: see line 353
\begin{claim}\label{clm:split-stil-covers}
	%When table $\sytb{H}$ is replaced by tables $\sytb{H}_1,\sytb{H}_2,\ldots,$ $\sytb{H}_m$ (in Alg.~\ref{alg:split}) then every concrete tree $\tree{T}$ covered by $\sytb{H}$ is covered by %exactly  one of $\sytb{H}_i$. 
	If a concrete tree $\tree{T}$ is covered by $\sytb{H}$ in line~\ref{line:H:alg:split}
	of Alg.~\ref{alg:split}, then when the alg. reaches line~\ref{line:endsplit:alg:split}, $\tree{T}$ is covered by $\sytb{H}_i$ for some $1\leq i \leq \ell$. 
\end{claim}

\begin{proof}[Proof Sketch]
	Let $W$ be the set of words obtained by concatenating a row $r\in R$ and a column $c\in C$ such that $M(r,c)$ is filled.
	Clearly a concrete tree $\tree{T}=\tuple{W,\tau}$ that was covered by $\sytb{H}$ agrees in every node of all branches but the branch $i_1i_2\ldots i_m$ with the respective entries in all $\sytb{H}_i$s.
	Suppose the concrete counterexample is $\w=\biword{\isym_1}{\osym_1} \biword{\isym_2}{\osym_2}\biword{\isym_3}{\osym_3}\ldots \biword{\isym_{k}}{\osym_{k}}$. Then the concrete tree $\tree{T}$ must disagree
	with $\w$ at some position, call the first one it disagrees with $i$. Then $\tree{T}$ is  covered by $\sytb{H}_i$.
\end{proof}

%\subsection*{Conflicting concretizations}

\commentout{
\begin{claim}\label{clm:embeded-tree-subsumed-by-symb-tree}
	If $\aut{C}$ is embedded in $\aut{S}$ then $\sema{\aut{C}}\subseteq \sema{\aut{S}}$.
\end{claim}
\begin{proof}[Proof sketch]
	Every run $q_0,q_1,\ldots, q_m$ of $\aut{C}$ on $\iw=i_1i_2\ldots i_m$ has a corresponding run $q'_0 q'_1 \ldots q'_m$ of $\aut{S}$ on $\iw$ in which 
	$q'_i\in h(q_i)$ and 
	if the corresponding sequence of outputs of $\aut{C}$ is $\ow=o_1o_2\ldots o_m$
	and the corresponding sequence of outputs of $\aut{S}$ is $\theta_1\theta_2\ldots \theta_m$ 
	then $o_i\in \theta_i$ for every $1\leq i\leq m$.
\end{proof}}

%\subsubsection{The Incompatibility Measure} %Compatible vs. Conflicting Trasducers}
In the body of the paper we defined the incompatibility measure with respect to a pair of transducers.
We provide here a definition that generalizes it for a set of transducers. Then we show that 
if a set of transducers are incompatible with respect to $U$, then there exists two transducers contained in $U$
that are incompatible.

Let $U$ be a target language and $\aut{S}_1,\ldots,\aut{S}_k$ symbolic transducers embedded in $U$.
A word $\iw=\isym_1\isym_2\ldots\isym_m\isym_{m+1}$ is said to \emph{witness the incompatibility} of $\aut{S}_1,\ldots,\aut{S}_k$  wrt. $U$ if 
given ${\given{\sema{\aut{S}_i}}{\iw}=\theta^i_1\theta^i_2\ldots\theta^i_m}$ 
and letting $\theta_j=\cup_{1\leq i \leq m}\theta^i_j$ for $1\leq j\leq m$,
the result of $\smq(\isym_1\isym_2\ldots\isym_m\isym_{m+1}\oplus \theta_1\theta_2\ldots\theta_m\qm )$ is $\emptyset$, and the result for any respective prefix is not $\emptyset$.
The transducers $\aut{S}_1,\ldots,\aut{S}_k$ are said to be \emph{incompatible} (or \emph{conflicting}) wrt. $U$  if there exists a word witnessing their incompatibility, otherwise
 they are said to be \emph{compatible}.

\begin{lemma}\label{lem:split-by-two}
    If $\aut{S}_1,\ldots,\aut{S}_k$ are incompatible wrt to $U$ as witnessed by $\iw$,
    then there exists $1\leq i \neq j \leq k$ s.t. $\aut{S}$ and $\aut{S}'$ contained in $U$ 
    that are incompatible wrt to $U$ as witnessed by $\iw$.
\end{lemma}

\begin{proof}[Proof sketch]
    Assume this is not the case. Then $\aut{S}_1$ and $\aut{S}_2$ are compatible.
    Therefore we can represent them by one symbolic transducer $\aut{S}_{12}$ (i.e. $\sema{\aut{S}_{12}}=\sema{\aut{S}_{1}}\cup\sema{\aut{S}_{2}}$), and
    $\aut{S}_{12},\aut{S}_3,\ldots,\aut{S}_k$ should still be incompatible wrt to $U$ with the same witness $\iw$.
    It follows that we can continue in the same fashion and represent pairs of transducers by one transducer
    until we can represent the original set of transducers by a pair of transducers that are still be incompatible wrt to $U$ with the same witness $\iw$.
    If we could have unite these two as well, then the original set would be compatible as well.
\end{proof}

The following claim asserts that a split occurs only if incompatibility
was discovered.

\begin{proposition}\label{prop:splitting}
	If the call $\smq(\sytb{H}(\iw)\cdot \biword{\isym}{\qm} )$ returns $\false$, 
	then there exists a pair of transducers $\aut{S}$ and $\aut{S}'$ contained in $U$ that are incompatible wrt. $U$
	and the word $\iw\isym$ witnesses their incompatibility.
\end{proposition}

\begin{proof}[Proof sketch]
    Let~$\iw=\isym_1\ldots \isym_m$, and
    let~$\theta_i=\sytb{H}(\isym_1\ldots \isym_i)$.
	It follows from the fact that  $\smq(\sytb{H}(\iw)\cdot \biword{\isym}{\qm} )$ returns $\false$
	and the definition of $\sytb{H}(\iw)$ that there exists no $\theta\subseteq \Gamma$ but $\theta=\emptyset$
	for which $\biword{\isym_1}{\theta_1}\biword{\isym_2}{\theta_2}\ldots\biword{\isym_m}{\theta_m}\biword{\isym}{\theta}\subseteq U$.
	Thus there must exits a set of finite/infinite states transducers embedded in $U$ that their incompatibility is witnessed by $\iw\isym$.
	It follows from Lemma~\ref{lem:split-by-two} that there exists $\aut{S}$ and $\aut{S}'$
	contained in $U$ 
    that are incompatible wrt to $U$ as witnessed by $\iw$.
\end{proof}